%% file: main.tex
\documentclass[journal,draftclsnofoot,onecolumn,12pt,twoside]{IEEEtran}
\usepackage{amsmath,amsfonts}
\usepackage{algorithm}
\usepackage{array}
\usepackage[caption=false,font=normalsize,labelfont=sf,textfont=sf]{subfig}
\usepackage{textcomp}
\usepackage{stfloats}
\usepackage{url}
\usepackage{verbatim}
\usepackage{graphicx}
\usepackage{cite}
\usepackage{xifthen}
\newcommand{\prob}[1][]{
\ifthenelse{\isempty{#1}}%
      {\ensuremath{P}}%
    {\ensuremath{P\left\(#1\right\)}}%
}
\newcommand{\vect}[1]{\boldsymbol{\mathrm{#1}}}
\newcommand{\mat}[1]{\boldsymbol{\mathrm{#1}}}
\newcommand{\tp}[1]{#1^{\mathrm{T}}}
\newcommand{\hr}[1]{#1^{\mathrm{H}}}
\newcommand{\inv}[1]{\left(#1\right)^{-1}}

\newcommand{\tr}[1]{\mathrm{tr}\left(#1\right)}
\newcommand{\diag}[1]{\mathrm{diag}\left(#1\right)}
\newcommand*{\inC}[1]{\in\mathbb{C}^{#1}}

\newcommand{\expt}[1]{\mathbb{E}\left\{#1\right\}}
\newcommand{\cn}[2]{\ensuremath{\mathcal{C}\mathcal{N}\left(#1,#2\right)}}
\newcommand{\cw}[3]{\ensuremath{\mathcal{C}\mathcal{W}\left(#1,#2,#3\right)}}
\newcommand{\icw}[3]{\ensuremath{{\mathcal{C}\mathcal{W}}^{-1}\left(#1,#2,#3\right)}}

\DeclareMathOperator*{\minimize}{minimize}
\DeclareMathOperator*{\argmax}{arg\,max}

\usepackage{pgfplots}
\pgfplotsset{compat=1.18}
\usepgfplotslibrary{fillbetween}
\usepackage{tikz}
\usepackage{hyperref}
\usepackage{acronym}
\usepackage{subfig}
\usetikzlibrary{patterns}
\usepackage{algpseudocode}
\usepackage{tikz-3dplot}
\usetikzlibrary{math}
\usepackage[shortlabels]{enumitem}
\usepackage{orcidlink}

\usepackage[acronym]{glossaries}
\input{abbr.tex}
\glsdisablehyper

\usepackage[switch]{lineno} 

\usepackage{amsthm}
\newtheoremstyle{mystyle}
  {}
  {}
  {\itshape}
  {}
  {\bfseries}
  {.}
  { }
  {}
\theoremstyle{mystyle}
\newtheorem{theorem}{Theorem}
\newtheorem{lemma}{Lemma}
\newtheorem{corollary}{Corollary}

\begin{document}

\title{Energy-Saving Precoder Design for\\ Narrowband and Wideband Massive MIMO}

\author{Emanuele Peschiera,
        Fran\c{c}ois Rottenberg,~\IEEEmembership{Member,~IEEE}
\thanks{The authors are with ESAT-DRAMCO, Campus Ghent, KU Leuven, 9000 Ghent, Belgium (email: \href{mailto:emanuele.peschiera@kuleuven.be}{emanuele.peschiera@kuleuven.be}).}}

\date{}



\maketitle

\begin{abstract}
In this work, we study massive \gls{mimo} precoders optimizing power consumption while achieving the users' rate requirements. We first characterize analytically the solutions for narrowband and wideband systems minimizing the \glspl{pa} consumption in low system load, where the per-antenna power constraints are not binding. After, we focus on the asymptotic wideband regime. The power consumed by the whole \gls{bs} and the high-load scenario are then also investigated. We obtain simple solutions, and the optimal strategy in the asymptotic case reduces to finding the optimal number of active antennas, relying on known precoders among the active antennas. Numerical results show that large savings in power consumption are achievable in the narrowband system by employing antenna selection, while all antennas need to be activated in the wideband system when considering only the \glspl{pa} consumption, and this implies lower savings. When considering the overall BS power consumption and a large number of subcarriers, we show that significant savings are achievable in the low-load regime by using a subset of the \gls{bs} antennas. While optimization based on transmit power pushes to activate all antennas, optimization based on consumed power activates a number of antennas proportional to the load.
\end{abstract}
\glsresetall

\begin{IEEEkeywords}
Massive MIMO, Power amplifiers, Precoding, Power consumption.
\end{IEEEkeywords}

\input{1_Introduction.tex}

\input{2_System_Model.tex}

\input{3_Precoders_Design_Wideband_System.tex}

\input{4_Narrowband_System.tex}

\input{5_Asymptotic_Wideband_System.tex}

\input{6_Simulation_Results.tex}

\input{7_Conclusion.tex}

\input{Appendices.tex}

\bibliographystyle{IEEEtran}
\bibliography{References}

\vfill

\end{document}

%% file: abbr.tex
\newacronym{3gpp}{3GPP}{3rd generation partnership project}
\newacronym{4g}{4G}{fourth-generation}
\newacronym{5g}{5G}{fifth-generation}
\newacronym{6g}{6G}{sixth-generation}
\newacronym{abp}{ABP}{authentication by personalisation}
\newacronym{aclr}{ACLR}{adjacent channel leakage ratio}
\newacronym{adc}{ADC}{analog-to-digital converter}
\newacronym{adr}{ADR}{adaptive data rate}
\newacronym{ag}{AG}{array gain}
\newacronym{ai}{AI}{artificial intelligence}
\newacronym{amam}{AM/AM}{amplitude modulation to amplitude modulation}
\newacronym{amp}{AMP}{approximate message passing}
\newacronym{ampm}{AM/PM}{amplitude modulation to phase modulation}
\newacronym{aoa}{AOA}{angle-of-arrival}
\newacronym{aod}{AOD}{angle-of-departure}
\newacronym{ap}{AP}{access point}
\newacronym{ar}{AR}{augmented reality}
\newacronym{arp}{ARP}{Antenna Reference Point}
\newacronym{asic}{ASIC}{application specific integrated circuit}
\newacronym{auv}{AUV}{autonomous underwater vehicle}
\newacronym{awgn}{AWGN}{additive white Gaussian noise}
\newacronym{bldc}{BLDC}{brushless DC}
\newacronym{bom}{BOM}{bill of materials}
\newacronym{bs}{BS}{base station}
\newacronym{bw}{BW}{bandwidth}
\newacronym{cad}{CAD}{channel activity detection}
\newacronym{cbm}{CBM}{condition based maintenance}
\newacronym{cc}{CC}{constant current}
\newacronym{ccdf}{CCDF}{complementary cumulative distribution function}
\newacronym{ccnn}{CCNN}{circular convolutional neural network}
\newacronym{cdf}{CDF}{cumulative distribution function}
\newacronym{cdrx}{CDRX}{connected mode DRX}
\newacronym{ce}{CE}{coverage enhancement}
\newacronym{ced}{CED}{cumulative energy density}
\newacronym{cf}{CF}{cell-free}
\newacronym{cfo}{CFO}{carrier frequency offset}
\newacronym{cots}{COTS}{commercial off-the-shelf}
\newacronym{cp}{CP}{cyclic prefix}
\newacronym{cpt}{CPT}{capacitive power transfer}
\newacronym{cpu}{CPU}{central-processing unit}
\newacronym{cqi}{CQI}{channel quality indicator}
\newacronym{cr}{CR}{coding rate}
\newacronym{crlb}{CRLB}{Cram\'er-Rao lower bound}
\newacronym{crs}{CRS}{cell reference signal}
\newacronym{cs}{CS}{compressed sensing}
\newacronym{csi}{CSI}{channel state information}
\newacronym{csp}{CSP}{contact service point}
\newacronym{css}{CSS}{chirp spread spectrum}
\newacronym{cv}{CV}{constant voltage}
\newacronym{dac}{DAC}{digital-to-analog converter}
\newacronym{daq}{DAQ}{data acquisition system}
\newacronym{dc}{DC}{direct current}
\newacronym{dcc}{DCC}{dynamic cooperation clustering}
\newacronym{de}{DE}{drain efficiency}
\newacronym{dl}{DL}{downlink}
\newacronym{dlc}{DLC}{Distributed Laser Charging}
\newacronym{dmimo}{D-MIMO}{distributed MIMO}
\newacronym{dpd}{DPD}{digital pre-distortion}
\newacronym{drx}{DRX}{Discontinuous Reception Mode}
\newacronym{e2e}{E2E}{end-to-end}
\newacronym{easa}{EASA}{European Union Aviation Safety Agency}
\newacronym{ecdf}{eCDF}{empirical cumulative distribution function}
\newacronym{ecsp}{ECSP}{edge computing service point}
\newacronym{edlc}{EDLC}{electrostatic double-layer capacitors}
\newacronym{edrx}{eDRX}{Extended Discontinuous Reception Mode}
\newacronym{ee}{EE}{energy efficiency}
\newacronym{egprs}{EGPRS}{Enhanced Data Rates for GSM Evolution}
\newacronym{eirp}{EIRP}{effective isotropic radiated power}
\newacronym{em}{EM}{electromagnetic}
\newacronym{embb}{eMBB}{enhanced Mobile Broadband}
\newacronym{en}{EN}{energy neutral}
\newacronym{eol}{EoL}{end of life}
\newacronym{epu}{EPU}{edge processing unit}
\newacronym{erp}{ERP}{effective radiated power}
\newacronym[plural=ESCs,firstplural=electronic speed controllers (ESCs)]{esc}{ESC}{electronic speed control}
\newacronym{etsi}{ETSI}{European Telecommunications Standards Institute}
\newacronym{evm}{EVM}{Error Vector Magnitude}
\newacronym{fa}{FA}{federation anchor}
\newacronym{fft}{FFT}{fast Fourier transform}
\newacronym{fh}{FH}{fronthaul}
\newacronym{fim}{FIM}{Fisher information matrix}
\newacronym{fom}{FoM}{figure of merit}
\newacronym{fpga}{FPGA}{field-programmable gate array}
\newacronym{fdma}{FDMA}{frequency-division multiple access}
\newacronym{gb}{GB}{grant-based}
\newacronym{gf}{GF}{grant-free}
\newacronym{gnb}{gNB}{Next Generation Node B}
\newacronym{gnn}{GNN}{graph neural network}
\newacronym{gnss}{GNSS}{global navigation satellite system}
\newacronym{gpclk}{GPCLK}{general purpose clock}
\newacronym{gprs}{GPRS}{General Packet Radio Services}
\newacronym{gps}{GPS}{Global Positioning System}
\newacronym{gpu}{GPU}{graphical processing unit}
\newacronym{gsm}{GSM}{Global System for Mobile Communications}
\newacronym{gwp}{GWP}{Global Warming Potential}
\newacronym{harq}{HARQ}{hybrid automatic repeat request}
\newacronym{hat}{HAT}{hardware attached on top}
\newacronym{hcs}{HCS}{human-centric services}
\newacronym{hpbm}{HPBM}{half power beam width}
\newacronym{HyMPRo}{HyMPRo}{Hybrid Multi-Path Routing algorithm}
\newacronym{ib}{IB}{in-band}
\newacronym{ibo}{IBO}{input back-off}
\newacronym{ic}{IC}{integrated circuit}
\newacronym{ict}{ICT}{information and communications technology}
\newacronym{id}{ID}{information decoding}
\newacronym{if}{IF}{intermediate frequency}
\newacronym{iid}{i.i.d.}{independently and identically distributed}
\newacronym{im}{IM}{intermodulation}
\newacronym{imu}{IMU}{inertial measurement unit}
\newacronym{iot}{IoT}{Internet of things}
\newacronym{ipt}{IPT}{inductive power transfer}
\newacronym{ipy}{IPY}{Interventions per Year}
\newacronym{ism}{ISM}{industrial, scientific and medical}
\newacronym{isp}{ISP}{internet service provider}
\newacronym{kkt}{KKT}{Karush--Kuhn--Tucker}
\newacronym{kpi}{KPI}{key performance indicator}
\newacronym{kvi}{KVI}{key value indicator}
\newacronym{lca}{LCA}{life cycle assessment}
\newacronym{lco}{LCO}{lithium cobalt oxide}
\newacronym{ldo}{LDO}{Low-dropout voltage regulator}
\newacronym{ldpc}{LDPC}{low-density parity-check}
\newacronym{lfp}{LFP}{lithium iron phosphate}
\newacronym{lic}{LIC}{lithium-ion capacitor}
\newacronym{liion}{Li-ion}{lithium-ion}
\newacronym{lipo}{LiPo}{lithium polymer}
\newacronym{lis}{LIS}{large intelligent surface}
\newacronym{llh}{LLH}{log-likelihood}
\newacronym{lmmse}{LMMSE}{least minimum mean square error}
\newacronym{lmo}{LMO}{lithium ion manganese oxide}
\newacronym{lora}{LoRa}{long range}
\newacronym{lorawan}{LoRaWAN}{long-range wide-area network}
\newacronym{los}{LoS}{line-of-sight}
\newacronym{lp}{LP}{linear programming}
\newacronym{lpt}{LPT}{laser power transfer}
\newacronym{lpwa}{LPWA}{Low Power Wide Area}
\newacronym{lpwan}{LPWAN}{low-power wide-area network}
\newacronym{lqi}{LQI}{link quality indicator}
\newacronym{lrelu}{LReLU}{leaky rectified linear unit}
\newacronym{lrt}{LRT}{likelihood-ratio test}
\newacronym{ls}{LS}{least squares}
\newacronym{lsfc}{LSFC}{large-scale fading component}
\newacronym{ltc}{LTC}{lithium thionyl chloride}
\newacronym{lte}{LTE}{long term evolution}
\newacronym{lto}{LTO}{lithium titanate}
\newacronym{m2m}{M2M}{machine to machine}
\newacronym{mac}{MAC}{Medium Access Control}
\newacronym{mcl}{MCL}{Maximum Coupling Loss}
\newacronym{mcs}{MCS}{modulation and coding scheme}
\newacronym{mcu}{MCU}{Microcontroller Unit}
\newacronym{mec}{MEC}{multi-access edge computing}
\newacronym{mems}{MEMS}{micro-electromechanical systems}
\newacronym{mf}{MF}{matched filter}
\newacronym{mimo}{MIMO}{multiple-input multiple-output}
\newacronym{miso}{MISO}{multiple-input single-output}
\newacronym{ml}{ML}{machine learning}
\newacronym{mlp}{MLP}{multilayer perceptron}
\newacronym{mmimo}{mMIMO}{massive MIMO}
\newacronym{mmse}{MMSE}{minimum mean square error}
\newacronym{mmtc}{mMTC}{massive machine-typed communication}
\newacronym{mmwave}{mmWave}{millimeter wave}
\newacronym{mpc}{MPC}{multipath component}
\newacronym{mppt}{MPPT}{maximum power point tracking}
\newacronym{mr}{MR}{maximum ratio}
\newacronym{mrc}{MRC}{maximum ratio combining}
\newacronym{mrc_em}{MRC}{maximum ratio combining}
\newacronym{mrt}{MRT}{maximum ratio transmission}
\newacronym{mse}{MSE}{mean square error}
\newacronym{mtc}{MTC}{Machine-Type Communication}
\newacronym{multi-rat}{Multi-RAT}{multiple radio access technology}
\newacronym{nb}{NB}{narrowband}
\newacronym{nbiot}{NB-IoT}{narrowband IoT}
\newacronym{nca}{NCA}{nickel cobalt aluminum}
\newacronym{nicd}{NiCd}{nikkel cadmium}
\newacronym{nimh}{NiMH}{nikkel metal hydride}
\newacronym{nlos}{NLoS}{non-line-of-sight}
\newacronym{nmc}{NMC}{nickel manganese cobalt}
\newacronym{nn}{NN}{neural network}
\newacronym{nnls}{NNLS}{non-negative least squares}
\newacronym{noma}{NOMA}{non-orthogonal multiple access}
\newacronym{nr}{NR}{New Radio}
\newacronym{ntp}{NTP}{network time protocol}
\newacronym{oai}{OAI}{OpenAirInterface} 
\newacronym{ofdm}{OFDM}{orthogonal frequency-division multiplexing}
\newacronym{ofdmim}{OFDM-IM}{OFDM with index modulation}
\newacronym{oob}{OOB}{out-of-band}
\newacronym{os}{OS}{operating system}
\newacronym{ota}{OTA}{over-the-air}
\newacronym{otaa}{OTAA}{over-the-air authentication}
\newacronym{p2p}{P2P}{point-to-point}
\newacronym{pa}{PA}{power amplifier}
\newacronym{pae}{PAE}{power-added efficiency}
\newacronym{papr}{PAPR}{peak-to-average power ratio}
\newacronym{pc}{PC}{pilot count}
\newacronym{pcb}{PCB}{printed circuit board}
\newacronym{pcg}{PCG}{power consumption gain}
\newacronym{pd}{PD}{powered device}
\newacronym{pdcch}{PDCCH}{physical downlink control channel}
\newacronym{pdp}{PDP}{power delay profile}
\newacronym{pdsch}{PDSCH}{physical downlink shared channel}
\newacronym{per}{PER}{packet error rate}
\newacronym{pg}{PG}{path gain}
\newacronym{pgd}{PGD}{proximal gradient descent}
\newacronym{pl}{PL}{path loss}
\newacronym{plr}{PLR}{packet loss ratio}
\newacronym{pll}{PLL}{phase-locked loop}
\newacronym{poe}{PoE}{power over Ethernet}
\newacronym{pps}{PPS}{pulse per second}
\newacronym{prbs}{PRBs}{Physical Resource Blocks}
\newacronym{psd}{PSD}{power spectral density}
\newacronym{pse}{PSE}{power sourcing equipment}
\newacronym{psm}{PSM}{power saving mode}
\newacronym{pss}{PSS}{primary synchronisation signal}
\newacronym{ptp}{PTP}{Precision Time Protocol}
\newacronym{pdr}{PDR}{packet delivery ratio}
\newacronym{ptw}{PTW}{paging time window}
\newacronym{qam}{QAM}{quadrature amplitude modulation}
\newacronym{qos}{QoS}{quality-of-service}
\newacronym{quadriga}{QuaDRiGa}{QUAsi Deterministic RadIo channel GenerAtor}
\newacronym{ra}{RA}{Random Access}
\newacronym{ran}{RAN}{radio access network}
\newacronym{rar}{RAR}{Random Access Response}
\newacronym{rat}{RAT}{radio access technology}
\newacronym{rbs}{RBS}{radio base station}
\newacronym{re}{RE}{radio element}
\newacronym[]{relu}{ReLU}{rectified linear unit}
\newacronym{rf}{RF}{radio frequency}
\newacronym{rfeh}{RFEH}{radio frequency energy harvesting}
\newacronym{rfic}{RFIC}{radio-frequency integrated circuit}
\newacronym{rfid}{RFID}{radio frequency identification}
\newacronym{rfpt}{RFPT}{radio frequency power transfer}
\newacronym{ris}{RIS}{reflective intelligent surface}
\newacronym{rms}{RMS}{root-mean-square}
\newacronym{rmse}{RMSE}{root-mean-square error}
\newacronym{ros}{ROS}{robot operating system}
\newacronym{rpi}{RPi}{Raspberry Pi}
\newacronym{rrc}{RRC}{Radio Resource Connection}
\newacronym{rreq}{RREQ}{route request packet}
\newacronym{rsrp}{RSRP}{Reference Signals Received Power}
\newacronym{rss}{RSS}{received signal strength}
\newacronym{rssi}{RSSI}{received signal strength indicator}
\newacronym{rtc}{RTC}{real time clock}
\newacronym{rtk}{RTK}{real time kinematics}
\newacronym{rts}{RTS}{ray tracing simulator}
\newacronym{rw}{RW}{RadioWeaves}
\newacronym{rx}{RX}{receiver}
\newacronym{rzf}{RZF}{regularized zero forcing}
\newacronym{sa}{SA}{synchronization anchor}
\newacronym{scfdma}{SCFDMA}{single-carrier frequency division multiple access}
\newacronym{sdg}{SDG}{Sustainable Development Goal}
\newacronym{sdn}{SDN}{software-defined network}
\newacronym{sdr}{SDR}{signal-to-distortion ratio}
\newacronym{sf}{SF}{spreading factor}
\newacronym{sfn}{SFN}{single frequency network}
\newacronym{sinr}{SINR}{signal-to-interference-plus-noise ratio}
\newacronym{siso}{SISO}{single-input single-output}
\newacronym{slam}{SLAM}{simultaneous localization and mapping}
\newacronym{slc}{SLC}{spatial leakage suppression}
\newacronym{smc}{SMC}{specular multipath component}
\newacronym{smps}{SMPS}{switched mode power supply}
\newacronym{sndr}{SNDR}{signal-to-noise-and-distortion ratio}
\newacronym{snidr}{SNIDR}{signal-to-noise-and-interference-and-distortion ratio}
\newacronym{snr}{SNR}{signal-to-noise ratio}
\newacronym{soc}{SoC}{state of charge}
\newacronym{ssb}{SSB}{synchronisation signal block}
\newacronym{ssd}{SSD}{solid state drive}
\newacronym{steam}{STEAM}{science, technology, engineering, the arts, and mathematics}
\newacronym{svd}{SVD}{singular value decomposition}
\newacronym{tau}{TAU}{tracking area update}
\newacronym{tcer}{TCER}{transported to consumed energy ratio}
\newacronym{tdd}{TDD}{time division duplexing}
\newacronym{tdoa}{TDOA}{time-difference-of-arrival}
\newacronym{toa}{TOA}{time-of-arrival}
\newacronym{tof}{ToF}{time-of-flight}
\newacronym{trp}{TRP}{Transmission Reception Point}
\newacronym{tsn}{TSN}{time-sensitive networking}
\newacronym{ttm}{TTM}{time to market}
\newacronym{ttn}{TTN}{The Things Network}
\newacronym{tx}{TX}{transmitter}
\newacronym{tsch}{TSCH}{time slotted channel hopping}
\newacronym{tdma}{TDMA}{time division multiple access}
\newacronym{uav}{UAV}{unmanned aerial vehicle}
\newacronym{udp}{UDP}{User Datagram Protocol}
\newacronym{ue}{UE}{user equipment}
\newacronym{ugv}{UGV}{unmanned ground vehicle}
\newacronym{uhd}{UHD}{USRP hardware driver}
\newacronym{uhf}{UHF}{ultra-high frequency}
\newacronym{ul}{UL}{uplink}
\newacronym{ula}{ULA}{uniform linear array}
\newacronym{upa}{UPA}{uniform planar array}
\newacronym{ura}{URA}{uniform rectangular array}
\newacronym{urllc}{URLLC}{ultra-reliable low-latency communications}
\newacronym{usrp}{USRP}{universal software radio peripheral}
\newacronym{uv}{UV}{unmanned vehicle}
\newacronym{uwb}{UWB}{ultrawideband}
\newacronym{vep}{VEP}{virtual edge platform}
\newacronym{vlc}{VLC}{visible light communication}
\newacronym{vlp}{VLP}{visible light positioning}
\newacronym{wb}{WB}{wideband}
\newacronym{wpt}{WPT}{wireless-power transfer}
\newacronym{wr}{WR}{white-rabbit}
\newacronym{wrsn}{WRSN}{wireless rechargeable sensor network}
\newacronym{wsn}{WSN}{wireless sensor network}
\newacronym{z3ro}{Z3RO}{zero third-order distortion}
\newacronym{zf}{ZF}{zero-forcing}
\newacronym{zmcscg}{ZMCSCG}{zero mean circularly symmetric complex Gaussian}

%% file: 1_Introduction.tex
\section{Introduction}
\label{sec:intro}

\subsection{Problem Statement}

\IEEEPARstart{T}{he} reduction of electrical energy consumption and carbon footprint is among the priorities of our society~\cite{EGD}. The \gls{ict} industry targets to reduce its greenhouse gas emissions by 45\% by 2030~\cite{ITU}. Moreover, decreasing the energy consumption of radio access networks lowers the operational expenses of operators, which have the largest impact on their bills~\cite{Andersson22}. Energy efficiency of mobile communications has been the object of research for more than a decade~\cite{Han11,Buzzi16,LPerez22}, but there is still room for improvement. Measurements~\cite{Huawei,Golard22} show that the power consumption of current wireless networks does not depend much on data traffic and barely decreases even when the traffic is close to zero. This creates opportunities for energy savings: improving the load dependence of the consumption by, e.g., implementing adaptive shutdown of the network components~\cite{Andersson22,Huawei,3gpp}. Recent studies highlight how energy-saving gains up to 48.2\% are achievable by only adapting the spatial transmission resources to the traffic~\cite{3gpp}.

The \gls{pa} is among the main contributors to the energy consumption of a \gls{bs}~\cite{Auer11}, even though the share of the baseband processing is increasing in \gls{5g} systems~\cite{ge17}. In \gls{ofdm}, the current broadband modulation format used in cellular systems, the \gls{pa} operates in linear regime to avoid signal distortion and out-of-band emissions. However, the \gls{pa} efficiency is maximal when operating in the saturation region, leading to a trade-off between system capacity and power consumption~\cite{Muneer20}. Despite this, the conventional massive \gls{mimo} precoders are typically found by optimizing a certain figure of merit at the receivers, i.e., rate, \gls{sinr}, or \gls{mse}, under a transmit power constraint~\cite{Bjornson14}. Indeed, the maximal transmit power is commonly fixed by regulations. This would be optimal --~in terms of \glspl{pa} consumption~-- if the \gls{pa} efficiency was fixed, but it actually depends on the output power level~\cite{Grebennikov05}, making power consumption and transmit power not linearly proportional. Moreover, the optimization of the energy performance of existing precoders is traditionally based on their \gls{ee}, i.e., the number of transferred bits per unit of energy~\cite{Chen11}. However, \gls{ee} maximization does not guarantee a given \gls{qos} and can lead to undesirable situations, e.g., low power consumption but small rates, or large rates but equally high power consumption. An alternative approach is to fix the \gls{qos} requirements and minimize the power consumption. In this way, one can optimize the transmission resources for a specific traffic level, and avoid allocating too many resources when not necessary~\cite{3gpp}. The above makes the trade-off between energy consumption and performance requirements transparent~\cite{Andersson22}.

\subsection{State-of-the-Art}

Considering a \gls{bs} equipped with $M$ antennas, $p_m$ is defined as the power at the output of antenna $m$. The total transmit power is then given by $p_{\mathrm{tx}} = \sum_{m=0}^{M-1}p_m$. A common assumption is that the consumed power by the \gls{pa} at antenna $m$ can be modeled as $\frac{p_m}{\eta}$, where $\eta$ is the fixed \gls{pa} efficiency. In this way, the total power consumed by the \glspl{pa} is linearly related to the total transmit power
\begin{equation}
\label{eq:p_cons_1}
    p_{\mathrm{PAs}}^\mathrm{(ideal)} = \sum_{m=0}^{M-1}\frac{p_m}{\eta} = \frac{p_{\mathrm{tx}}}{\eta}.
\end{equation}
However, (\ref{eq:p_cons_1}) represents an ideal model. A more realistic expression of the efficiency of the PA $m$ is $\eta_m = \eta_{\mathrm{sat}}\left(\frac{p_m}{p_{\mathrm{sat}}}\right)^{1/2}$, where $p_{\mathrm{sat}}$ is the \gls{pa} saturation power and $\eta_{\mathrm{sat}}$ is the maximal \gls{pa} efficiency, achieved when $p_m = p_\mathrm{sat}$~\cite{Grebennikov05}. This expression is accurate for class B amplifiers~\cite{He11}. Considering also an industrial \gls{pa} for \gls{5g} \gls{mmimo}~\cite{Qorvo}, the square-root behavior proves to be suited to model the efficiency when far from saturation (e.g., considering a $10$ dB back-off). A back-off of $10$ dB is typically used in \gls{ofdm} systems to operate in the linear regime~\cite{Auer11}, given the high \gls{papr}. Under this model, the total power consumed by the \glspl{pa} becomes
\begin{equation}
\label{eq:p_cons_2}
    p_{\mathrm{PAs}} = \sum_{m=0}^{M-1}\frac{p_m}{\eta_m} = \frac{p_{\mathrm{sat}}^{1/2}}{\eta_{\mathrm{sat}}}\sum_{m=0}^{M-1}p_m^{1/2}.
\end{equation}
Despite the extensive research in \gls{mmimo} signal processing, few works have considered (\ref{eq:p_cons_2}) in the precoders design, especially for a multicarrier system. In~\cite{Persson13,Persson14}, the authors analyzed the cases of \gls{miso} and \gls{mimo} with orthogonal channels. In the \gls{miso} case, the problem is formulated as minimization of (\ref{eq:p_cons_2}) subject to a rate constraint and per-antenna power constraints. The solution corresponds to antenna selection, saturating the antennas with the largest channel gains and not activating the remaining antennas. The application to point-to-point \gls{mimo} is studied in~\cite{Cheng19}, but the multi-user \gls{mimo} system is not investigated in detail. The work in~\cite{Cheng15} assumes the general form of the precoder (\gls{mrt} and \gls{zf} schemes) and finds a posteriori the system parameters (number of \gls{bs} antennas, number of users, etc.) that minimize (\ref{eq:p_cons_2}). They optimize the number of utilized \gls{bs} antennas in a low-traffic regime, as in the high-traffic situation their choice is to always keep all the \gls{bs} antennas turned on to guarantee the coverage in the cell. The solution (when considering only the power consumed by the \glspl{pa}) is to activate all the \gls{bs} antennas, in contrast to~\cite{Persson13, Persson14}. The reason is that~\cite{Persson13, Persson14} consider the minimization of (\ref{eq:p_cons_2}) from the precoding problem definition, not assuming any particular precoding scheme. When other \gls{bs} components are included in the power consumption model, the optimal number of \gls{bs} antennas to be utilized decreases. 

Similarly to~\cite{Cheng15}, several studies have investigated the whole \gls{bs} power consumption for energy-efficient \gls{mmimo} communications. A power consumption model for a \gls{mmimo} \gls{bs} is developed in~\cite{massivemimobook}, where the total consumed power is the sum of the contributions from the \glspl{pa} and from the \gls{bs} circuits, plus a static consumption term. The circuit power consumption considers the radio frequency transceivers, the baseband processing, the backhauling, the cooling, and the power supply, and typically scales linearly with $M$~\cite{massivemimobook}. In~\cite{Senel19}, the authors jointly minimize the total transmit power and the circuit power consumption for \gls{mrt} and \gls{zf} precoders. Differently from~\cite{Cheng15}, they consider different traffic situations and they show that the number of activated antennas increases as a function of the traffic.

Remarkably, the multicarrier regime has not been investigated in enough detail in the literature. When considering a fixed \gls{pa} efficiency, the definition (\ref{eq:p_cons_1}) allows one to decouple and optimally solve the precoder design problem per-subcarrier. When optimizing the expression (\ref{eq:p_cons_2}), instead, the problem is coupled between subcarriers as the PA efficiency depends on the total power at its input. The work~\cite{Cheng15} considers (\ref{eq:p_cons_2}) in a \gls{mmimo} system assuming a uniform power among the antennas, given that each antenna serves a large number of subcarriers. Nonetheless, the form of the precoder is fixed a priori and this constitutes a limit in the analysis. The authors in~\cite{Cheng19} consider the case of point-to-point \gls{mimo} and optimize the sum-capacity subject to consumed power constraints and per-antenna power constraints. The expression of the consumed power takes into account the multicarrier nature of the system. Their conclusion is in line with ours (i.e., all the \gls{bs} antennas are equally good, then random antenna selection performs well), but the analysis of the \gls{mmimo} scenario and further elaboration on how to optimize the parameters of the system are missing.

Available works have considered the maximization of the \gls{ee} in \gls{mimo} precoding~\cite{Xu13,Bjornson15}, in some cases also considering a non-fixed \gls{pa} efficiency~\cite{Dong17,Hossain18}. As previously mentioned, our approach does not focus on the \gls{ee}, but fixes the rate requirements and derives the precoder that minimizes the power consumption. Indeed, conventional precoders can be obtained as solutions to optimization problems that minimize the transmit power subject to \gls{qos} constraints~\cite{Bjornson14}. Therefore, one can optimize the transmission resources for an instance of \gls{qos} requirements, and repeat this process once the requirements change.

\subsection{Contributions}
In this paper, we propose \gls{mmimo} precoders that minimize the consumed power under a zero inter-user interference constraint, starting from a general multi-user multicarrier system. The study is an extension of our previous work~\cite{Peschiera22}, where we only considered the \glspl{pa} consumption and the single-carrier scenario (referred to as narrowband to underline that we consider the channel to be frequency flat). Our main contributions are:
\begin{itemize}
\item We derive closed-form solutions to the problem of minimizing the \glspl{pa} consumption in wideband systems (Section~\ref{sec:MC_system}), where the powers at the \gls{bs} antennas are computed via iterative fixed point algorithm. We focus on the low-load scenario, in which the antenna output powers are lower than the maximal \gls{pa} power and the per-antenna power constraints are not binding.
\item We characterize analytically the solutions to the same problem in narrowband systems and low load (Section~\ref{sec:SC_system}), and we show that the single-user solution does not require an iterative method to be retrieved. 
\item We analyze the asymptotic behavior of wideband systems (Section~\ref{sec:asymptotic_MC}), solving the problem of minimizing the \gls{bs} consumption subject to per-antenna power constraints. We derive the optimal number of active antennas, which depends on the traffic load.
\end{itemize}
In Section~\ref{sec:sim_results}, the performance of the proposed solutions are evaluated numerically, and their complexity is discussed. Section~\ref{sec:conclusion} concludes the paper.

\subsection*{Notations}
Vectors and matrices are denoted by bold lowercase $\vect{a}$ and uppercase letters $\mat{A}$, respectively. The superscripts $(\cdot)^*$, $\tp{(\cdot)}$ and $\hr{(\cdot)}$ indicate complex conjugate, transpose, and conjugate transpose operations. The symbols $\tr{\cdot}$ and $\expt{\cdot}$ indicate the trace and the expectation operators. The notation $\diag{\vect{a}}$ refers to a diagonal matrix whose $k$-th diagonal entry is equal to the $k$-th element of $\vect{a}$. A diagonal matrix associated with the vector $\vect{a}$ is indicated by $\mat{D}_{\vect{a}}$. The identity matrix of size $K$ is $\mat{I}_K$. The $(i,j)$-th element of $\mat{A}$ is indicated by $\left[\mat{A}\right]_{i,j}$. The notation $\mathcal{CN}(\mu,\sigma^2)$ stands for a complex normal distribution with mean $\mu$ and variance $\sigma^2$. A complex Wishart distribution with mean $\mat{M}$ and $n$ degrees of freedom is indicated by $\mathcal{CW}(\mat{M},n,p)$, while a complex inverse Wishart distribution with mean $\mat{M}\inC{p\times p}$ and $n$ degrees of freedom is denoted by $\mathcal{CW}^{-1}(\mat{M},n,p)$. The symbol $\delta_{i,j}$ is the Kronecker delta function. If $f\colon \mathcal{X}\mapsto\mathbb{R}$ and $x\in\mathcal{X}$, $\lfloor x \rceil$ selects, among the two closest integers to $x$, the one associated to the minimum value of $f$. The matrix $\mat{A}^{1/2}$ is the only Hermitian positive semidefinite matrix satisfying $\mat{A} = \mat{A}^{1/2}\mat{A}^{1/2}$, $\mat{A}^{n}$ is the multiplication of $n$ $\mat{A}$ matrices and $\mat{A}^{-n}$ is the multiplication of $n$ $\mat{A}^{-1}$ matrices, for $n\in\mathbb{N}^+$. If $z=x+\jmath y$, with $x, y \in \mathbb{R}$, the Wirtinger derivative is defined as $\frac{\partial}{\partial z^*} = \frac{1}{2}\left(\frac{\partial}{\partial x}+\jmath\frac{\partial}{\partial y}\right)$. $\mathcal{O}(\cdot)$ stands for the big O notation.

%% file: 2_System_Model.tex
\section{System Model}
\label{sec:sys_model}

\subsection{Transmission Model}

We consider the downlink of a single-cell \gls{mmimo} \gls{ofdm} system with $Q$ subcarriers. The $M$ antennas at the \gls{bs} serve $K$ single-antenna users using space-division multiplexing. The transmitted symbol for the user $k$ at the subcarrier $q$ is denoted by $s_{k,q}$, and the symbols are uncorrelated and of unit variance, i.e., $\expt{s_{k,q}s_{k',q'}^*} = \delta_{k,k'}\delta_{q,q'}$. Indicating with $w_{m,k,q}$ the precoding coefficient for antenna $m$, user $k$ and subcarrier $q$, the precoded signal in the frequency domain at antenna $m$ and subcarrier $q$ is
\begin{equation}
\label{eq:precoded_signal}
    x_{m,q} = \sum_{k=0}^{K-1}w_{m,k,q}s_{k,q}.
\end{equation}
We assume the cyclic prefix to be sufficiently long and the channel to be time-invariant over an \gls{ofdm} symbol period so that it can be considered flat at each subcarrier. Defining $h_{k,m,q}$ as the channel coefficient between user $k$ and antenna $m$ at subcarrier $q$, the signal received by the user $k$ at the subcarrier $q$ is
\begin{equation}
\label{eq:received_signal}
    r_{k,q} = \sum_{m=0}^{M-1}h_{k,m,q}\sum_{k'=0}^{K-1}w_{m,k',q}s_{k',q} + \nu_{k,q}
\end{equation}
where $\nu_{k,q} \sim \cn{0}{\sigma_\nu^2}$ represents the thermal noise, which is assumed \gls{iid} among subcarriers and users.

The above description represents a wideband multi-user scenario. The particular cases of narrowband and single-user systems are obtained for $Q=1$ and $K=1$, respectively. When we vary the number of subcarriers, the channel is assumed to remain flat at the subcarrier level. The underlying assumption is that the subcarrier spacing is fixed (i.e., the system bandwidth will decrease as a function of $Q$). The interest in analyzing a narrowband system is motivated, e.g., by the inclusion of services as \gls{nbiot} in the \gls{5g} standard~\cite{Dahlman20}.

\subsection{Power Consumption Model}

Considering the precoded signal in (\ref{eq:precoded_signal}), the transmit power at antenna $m$ equals
\begin{equation}
\label{eq:power_per-antenna}
    p_m = \sum_{q=0}^{Q-1}\expt{\left|x_{m,q}\right|^2} = \sum_{k=0}^{K-1}\sum_{q=0}^{Q-1}\left|w_{m,k,q}\right|^2.
\end{equation}
Consequently, the total transmit power at the \gls{bs} is given by
\begin{equation}
\label{eq:total_tx_power}
    p_{\mathrm{tx}} = \sum_{m=0}^{M-1}\sum_{q=0}^{Q-1}\expt{\left|x_{m,q}\right|^2} = \sum_{m=0}^{M-1}\sum_{k=0}^{K-1}\sum_{q=0}^{Q-1}\left|w_{m,k,q}\right|^2.
\end{equation}
Moreover, following the model in (\ref{eq:p_cons_2}), the total power consumed by the \glspl{pa} is given by
\begin{equation}
\label{eq:total_consumed_power}
    p_\mathrm{PAs} = \underbrace{\frac{p_{\mathrm{max}}^{1/2}}{\eta_{\mathrm{max}}}}_{\alpha}\sum_{m=0}^{M-1}\left(\sum_{k=0}^{K-1}\sum_{q=0}^{Q-1}\left|w_{m,k,q}\right|^2\right)^{1/2}.
\end{equation}
Here and in the following, we consider $p_\mathrm{max} = \frac{p_\mathrm{sat}}{\mathrm{BO}}$ as the maximal \gls{pa}'s power, where $\mathrm{BO}$ is the back-off (e.g., $\mathrm{BO}=10$), and $\eta_\mathrm{max}$ as the \gls{pa} efficiency when $p_m=p_\mathrm{max}$. This choice is justified by the fact that a \gls{mimo} \gls{ofdm} system must operate in the linear regime.\footnote{Reducing the back-off would allow us to achieve higher \glspl{pa} efficiencies, but would also require to introduce non-linear distortions in the signal model and, e.g., an additional constraint forcing the distortion to zero~\cite{Rottenberg23}, rendering the problems at hand more difficult to be solved.}

Expression (\ref{eq:total_consumed_power}) only considers the \glspl{pa} contribution, while other \gls{bs} components are also consuming. A simple, though appropriate, model of the \gls{bs} power consumption is
\begin{equation}
\label{eq:total_consumed_power_BS}
    p_{\mathrm{BS}} = p_\mathrm{PAs}+p_{\mathrm{fix}}+\mathcal{C}M_\mathrm{a}
\end{equation}
where $p_{\mathrm{fix}}$ is the static power consumption, $\mathcal{C}$ is a positive linear scaling constant, and $M_\mathrm{a}\leq M$ is the number of active antennas (i.e., associated with a non-zero transmit power $p_m$). The term $\mathcal{C}M_\mathrm{a}$ accounts for the non-static contribution of the components other than the \glspl{pa}, such as the transceiver chains and the baseband processing. A \gls{rf} transceiver chain comprises \glspl{dac}, \glspl{adc}, filters, and mixers, and its consumption is either zero (non-active antenna) or equal to a constant value (active antenna). For common linear precoding schemes, the power consumption of the baseband unit is also linear in the number of active antennas~\cite{massivemimobook}. 
It is worth pointing out that hybrid precoding strategies, using less \gls{rf} chains than digital streams, enable the reduction of the \gls{bs} consumption~\cite{Gao16}. However, this paper focuses on fully digital architectures, which are typically used in sub-6 GHz systems.

\subsection{Assumptions}

Some assumptions are made in certain sections of the paper:

    $\mathbf{(As1)}$: low-load scenario, $p_\mathrm{max} \to \infty$.
\\We point out that, by letting $p_\mathrm{max}$ go to infinity, we practically mean that the constraints on the maximal power per antenna are not binding, as it usually happens in a low-load scenario. Indeed, if we define the load at antenna $m$ as the ratio $\frac{p_m}{p_\mathrm{max}}$, in a low-load scenario (i.e., with few users or low path loss) we expect that $p_m \ll p_\mathrm{max} \; \forall m$.
    
    $\mathbf{(As2)}$: asymptotic wideband regime, $Q \to \infty$.
\\Also in this case, we underline that the results obtained by letting $Q$ go to infinity apply also for finite values of $Q$.

    $\mathbf{(As3)}$: uncorrelated Rayleigh fading, $\mat{H}_q = \mat{D}_{\vect{\beta}}^{1/2}\mat{G}_q$,
\\where $\mat{H}_q\inC{K\times M}$ is the channel matrix at subcarrier $q$, $\mat{D}_{\vect{\beta}} = \diag{\beta_0,\dotsc,\beta_{K-1}}$ is the large-scale fading matrix and $\beta_k$ is the large-scale fading coefficient of user $k$. The matrix $\mat{G}_q$ models the small-scale fading and is composed of \gls{iid} elements, where the single entry is $g_{k,m,q} \sim \cn{0}{1}$. The \gls{iid} assumption is made over space and not over frequency, therefore the channels of different subcarriers can be correlated.

%% file: 3_Precoders_Design_Wideband_System.tex
\section{Precoders Design for Wideband System}
\label{sec:MC_system}

\subsection{Traditional Transmit Power Solution}

Denoting by $\gamma_k$ the target \gls{sinr} for the user $k$, the precoder minimizing the total transmit power is the solution to
\begin{equation}
\label{eq:ZF_ptx}
\begin{aligned}
    & \minimize_{\{w_{m,k,q}\}}  \quad  	&& 	p_{\mathrm{tx}} = \sum_{m=0}^{M-1}p_m \\
    & \mathrm{subject\ to} \quad		    && 	\mat{H}_q\mat{W}_q = \tilde{\mat{D}}_{\vect{\gamma}}^{1/2}\sigma_\nu \;\; \forall q, \\
    &                                       &&   p_m \leq p_{\mathrm{max}} \;\; \forall m
\end{aligned}
\end{equation}
where $p_m$ is given by (\ref{eq:power_per-antenna}), $\mat{H}_q \in \mathbb{C}^{K \times M}$ and $\mat{W}_q \in \mathbb{C}^{M \times K}$ are the channel and precoding matrices at subcarrier $q$, respectively, $\tilde{\mat{D}}_{\vect{\gamma}} = \diag{\frac{\gamma_0}{Q},\dotsc,\frac{\gamma_{K-1}}{Q}}$ contains the target users' \glspl{sinr} normalized with respect to $Q$,\footnote{This \gls{sinr} normalization has a practical meaning: it prevents that the user data rate goes to infinity as the number of subcarriers grows.} and $\sigma_\nu$ is the noise standard deviation. The \gls{sinr} at each subcarrier for user $k$ is then $\frac{\gamma_k}{Q}$. The first constraint is a \gls{zf} constraint, forcing the inter-user interference to zero. In the majority of the next sections, we rely on $\mathbf{(As1)}$ and we will relax the maximal power constraints.

The problem at hand is convex and differentiable. It can be solved by, e.g., the Lagrange multipliers method. Under $\mathbf{(As1)}$, the solution at the subcarrier $q$ is given by
\begin{equation}
\label{eq:ZF_ptx_sol}
    \mat{W}_q = \hr{\mat{H}_q}\left(\mat{H}_q\hr{\mat{H}_q}\right)^{-1}\tilde{\mat{D}}_{\vect{\gamma}}^{1/2}\sigma_\nu
\end{equation}
which corresponds to the per-subcarrier \gls{zf}~\cite{spencer04}. Note that we consider instantaneous \gls{csi} to be available at the \gls{bs} for the computation of the precoding matrices.

\subsection{PAs Consumed Power Solution}

The precoder minimizing the total power consumed by the \glspl{pa} is instead found by solving the following problem
\begin{equation}
\label{eq:ZF_pcons}
\begin{aligned}
    & \minimize_{\{w_{m,k,q}\}}  \quad  	&& 	p_\mathrm{PAs} = \alpha\sum_{m=0}^{M-1}p_m^{1/2} \\
    & \mathrm{subject\ to} \quad		    && 	\mat{H}_q\mat{W}_q = \tilde{\mat{D}}_{\vect{\gamma}}^{1/2}\sigma_\nu \;\; \forall q, \\
    &                                       &&   p_m \leq p_{\mathrm{max}} \;\; \forall m.
\end{aligned}
\end{equation}
As previously assumed to derive the conventional~\gls{zf}, we consider $\mathbf{(As1)}$, thus neglecting the constraints on $p_m$. With this assumption, the transmit power at the individual antennas could become large. In practice, one can consider a rescaling of the rows of the precoding matrices corresponding to the antennas that are exceeding the maximum allowed power value, with the consequent reduction of the users' \gls{sinr} requirements. 
We focus on \gls{pa} consumption because (i) the \gls{pa} is usually one of the main drivers of the \gls{bs} consumption~\cite{Auer11}, and (ii) the addition of other consumption terms would render the problem less tractable. However, in Section~\ref{sec:asymptotic_MC} we will show how the asymptotic analysis allows us to consider the \gls{bs} consumption while remaining very accurate in practical situations with a finite realistic number of subcarriers.

The solution to the above problem can be found by, e.g., applying the Lagrange multipliers method, and is given in Theorem~\ref{th:ZF_pcons_sol}.
\begin{theorem}
\label{th:ZF_pcons_sol}
Under $(\mathbf{As1})$, the solution to problem (\ref{eq:ZF_pcons}) for the subcarrier $q$ is given by
\begin{equation}
\label{eq:ZF_pcons_sol}
    \mat{W}_q = \mat{D}_{\vect{p}}^{1/2}\hr{\mat{H}_q}\left(\mat{H}_q\mat{D}_{\vect{p}}^{1/2}\hr{\mat{H}}_q\right)^{-1}\tilde{\mat{D}}_{\vect{\gamma}}^{1/2}\sigma_\nu
\end{equation}
where $\mat{D}_{\vect{p}} = \mathrm{diag}\left(p_0,\dotsc,p_{M-1}\right) = \mathrm{diag}(\vect{p})$ is the matrix containing the transmit powers at the \gls{bs} antennas, which are found by solving the fixed point equations
\begin{equation}
\label{eq:MC-MU_fpe}
    p_m = \sum_{k=0}^{K-1}\sum_{q=0}^{Q-1}\left|\left[\mat{D}_{\vect{p}}^{1/2}\hr{\mat{H}_q}\left(\mat{H}_q\mat{D}_{\vect{p}}^{1/2}\hr{\mat{H}_q}\right)^{-1}\tilde{\mat{D}}_{\vect{\gamma}}^{1/2}\sigma_\nu\right]_{m,k}\right|^2
\end{equation}
for $m = 0,\dotsc,M-1$.
\end{theorem}
\begin{proof}
See Appendix~\ref{proof1}.
\end{proof}
The system of fixed point equations can be solved through the fixed point iteration method, which is given in Algorithm~\ref{alg:fpim}. Once $\mat{D}_{\vect{p}}$ is known, it can be substituted in (\ref{eq:ZF_pcons_sol}) to find the precoding matrices.

\begin{algorithm}[h!]\linespread{1}
	\caption{Fixed point method to retrieve the powers at the \gls{bs} antennas}\label{alg:fpim}
	\small 
	\begin{algorithmic}
		\Require $\varepsilon, I_\mathrm{max}$ \Comment{Set tolerance and max.\ number of iterations}
		\State $i \gets 0$
		\For{$m\in\{0,\dotsc,M-1\}$} \Comment{Init.\ powers at the \gls{bs} antennas}
			\State $p_m^{(0)} \gets 1$ 
			\State $p_m^{(-1)} \gets \infty$ 
		\EndFor
		
		\While{$i<I_\mathrm{max} \textbf{ and } \max_m\left\{|p_m^{(i)}-p_m^{(i-1)}|\right\}>\varepsilon$}
			\State $\mat{D}_{\vect{p}} \gets \diag{p_0^{(i)},\dotsc,p_{M-1}^{(i)}}$
			\For{$m\in\{0,\dotsc,M-1\}$} \Comment{Update powers at the \gls{bs} antennas}
				\State $p_m^{(i+1)} \gets (\ref{eq:MC-MU_fpe})$
			\EndFor	
		\State $i\gets i+1$
		\EndWhile
	\end{algorithmic}
\end{algorithm}

Let us now analyze the single-user case.
\begin{corollary}
\label{co:mc-su}
Under $\mathbf{(As1)}$, the particularization of Theorem~\ref{th:ZF_pcons_sol} when $K=1$ gives the solution for the subcarrier $q$
\begin{equation}
\label{eq:eq:ZF_pcons_sol_SU}
    \mat{w}_q = \frac{\sigma_\nu\gamma^{1/2}}{Q^{1/2}}\frac{1}{\sum_{m'=1}^{M-1}|h_{m',q}|^2p_{m'}^{1/2}}\mat{D}_{\vect{p}}^{1/2}\vect{h}^*_q
\end{equation}
where $\vect{w}_q \in \mathbb{C}^{M \times 1}$ and $\vect{h}_q \in \mathbb{C}^{M \times 1}$ are the precoding and channel vectors at the subcarrier $q$, respectively. The fixed point equations in the powers per antenna are
\begin{equation}
\label{eq:MC-SU_fpe}
    p_m = \frac{\sigma_\nu^2\gamma}{Q} p_m\sum_{q=0}^{Q-1}\frac{\left|h_{m,q}\right|^2}{\left(\sum_{m'=0}^{M-1}\left|h_{m',q}\right|^2p_{m'}^{1/2}\right)^2}
\end{equation}
for $m=0,\dotsc,M-1$.
\end{corollary}

In \gls{los} channels, the problem simplifies and this allows us to draw interesting insights.
\begin{corollary}
Under $\mathbf{(As1)}$, in pure \gls{los} channels (i.e., characterized by $|h_{m,q}|=1 \; \forall m,q$), and considering at least one $p_m$ different from zero, (\ref{eq:MC-SU_fpe}) reduces to
\begin{equation}
\label{eq:MC-SU_LOS}
    \sum_{m=0}^{M-1}p_{m}^{1/2} = \sigma_\nu\gamma^{1/2}.
\end{equation}
\end{corollary}
\noindent Every transmit power allocation satisfying~(\ref{eq:MC-SU_LOS}) is, therefore, an optimal solution to the wideband single-user problem in the low-load case, and thus also the narrowband single-user solution. One has now freedom in the precoder design. All the power can be allocated to antenna $m$, which will give $p_m = \sigma_\nu^2\gamma$. A uniform allocation is another possible choice, with $p_m = p= \frac{\sigma_\nu^2\gamma}{M^2}$. Random allocations are also doable, as long as they fulfill equation~(\ref{eq:MC-SU_LOS}). In practice, however, it is better to activate the minimum number of antennas, so that the \gls{rf} chains of the non-active antennas can be turned off to save power.

%% file: 4_Narrowband_System.tex
\section{Narrowband System}
\label{sec:SC_system}

\subsection{Traditional Transmit Power Solution}

When $Q=1$ the system reduces to a narrowband one. The subcarrier index $q$ is then discarded. The conventional solution, minimizing the total transmit power, corresponds to
\begin{equation}
\label{eq:ZF_ptx_sol_SC}
    \mat{W} = \hr{\mat{H}}\left(\mat{H}\hr{\mat{H}}\right)^{-1}\mat{D}_{\vect{\gamma}}^{1/2}\sigma_\nu
\end{equation}
where $\mat{D}_{\vect{\gamma}} = \diag{\gamma_0,\dotsc,\gamma_{K-1}}$, while $\mat{H}\inC{K\times M}$ and $\mat{W}\inC{M\times K}$ are the single-carrier channel and precoding matrices, respectively.

\subsection{PAs Consumed Power Solution}

\begin{corollary}
Under $\mathbf{(As1)}$, it directly follows from Theorem~\ref{th:ZF_pcons_sol} that the precoding matrix for a narrowband system is
\begin{equation}
\label{eq:ZF_pcons_sol_SC}
    \mat{W} = \mat{D}_{\vect{p}}^{1/2}\hr{\mat{H}}\left(\mat{H}\mat{D}_{\vect{p}}^{1/2}\hr{\mat{H}}\right)^{-1}\mat{D}_{\vect{\gamma}}^{1/2}\sigma_\nu.
\end{equation}
The fixed point equation for the antenna $m$ is given by
\begin{equation}
\label{eq:SC-MU_fpe}
    p_m = \sum_{k=1}^{K}\left|\left[\mat{D}_{\vect{p}}^{1/2}\hr{\mat{H}}\left(\mat{H}\mat{D}_{\vect{p}}^{1/2}\hr{\mat{H}}\right)^{-1}\mat{D}_{\vect{\gamma}}^{1/2}\sigma_\nu\right]_{m,k}\right|^2.
\end{equation}
\end{corollary}
The solution can be retrieved in the same way as the wideband system, first computing the powers per antenna and substituting them in (\ref{eq:ZF_pcons_sol_SC}).

The single-user case can now be discussed.
\begin{corollary}
Under $\mathbf{(As1)}$, by assuming different channel gains among the antennas and by defining $\hat{m}=\argmax_{m} |h_m|$, the particularization of Corollary~\ref{co:mc-su} when $Q=1$ gives the single-user solution
\begin{equation}
w_m = 
\begin{cases}
   \sigma_\nu\gamma^{1/2}\frac{1}{|h_{m}|^2}h_m^* & \mathrm{if}\ m = \hat{m}\\
    0 & \mathrm{otherwise}
\end{cases}.
\end{equation}
All the power is therefore allocated to the antenna with the largest channel gain.
\end{corollary}
\begin{proof}
From Corollary~\ref{co:mc-su}, when $Q=1$ we obtain
\begin{equation}
    \mat{w} = \sigma_\nu\gamma^{1/2}\frac{1}{\sum_{m'=0}^{M-1}|h_{m'}|^2p_{m'}^{1/2}}\mat{D}_{\vect{p}}^{1/2}\vect{h}^*
\end{equation}
where $\vect{w} \in \mathbb{C}^{M \times 1}$ and $\vect{h} \in \mathbb{C}^{M \times 1}$ are the single-carrier precoding and channel vectors, respectively. The fixed point equation for the antenna $m$ is
\begin{equation}
\label{eq:fpe_SC-SU_2}
    \sum_{m'=0}^{M-1}|h_{m'}|^2p_{m'}^{1/2} = \sigma_\nu\gamma^{1/2}|h_m|.
\end{equation}
For all the activated antennas, there is a fixed point equation. However, when the channel gains $|h_m|$ are different, (\ref{eq:fpe_SC-SU_2}) cannot be solved since the left-hand side does not depend on $m$, while the right-hand side does. This actually shows that the optimal solution corresponds to using only one antenna, the one with the strongest channel gain~\cite{Peschiera22}.\footnote{When only one antenna is activated, the consumed power equals $\frac{\sigma_\nu\gamma^{1/2}}{|h_m|}$, which is minimized when using the antenna with the highest channel gain.} This solution makes sense since the maximal per-antenna power constraint is not considered here. 
\end{proof}

If several antennas share the same maximal channel gains, they can be all activated while still solving equation~(\ref{eq:fpe_SC-SU_2}). Equation~(\ref{eq:fpe_SC-SU_2}) reduces to equation~(\ref{eq:MC-SU_LOS}) among the antennas sharing the same maximal channel gain. In pure \gls{los} channels, we obtain again~(\ref{eq:MC-SU_LOS}), this time among all the antennas. In both cases, one has freedom of choice in the precoder design.

When the per-antenna power constraints are active, instead, the optimal strategy is to progressively saturate the antennas with the highest channel gains until the \gls{qos} is achieved~\cite{Persson13}.

%% file: 5_Asymptotic_Wideband_System.tex
\section{Asymptotic Wideband System}
\label{sec:asymptotic_MC}

We now assume a large value of $Q$, which makes sense given that in \gls{4g} and \gls{5g} systems the number of active subcarriers varies from $72$ to $1320$ and from $132$ to $3300$, respectively~\cite{Dahlman20}. The asymptotic results are validated numerically, and we show how they prove to be accurate even for a finite and relatively low number of subcarriers. The closed-form expressions of the consumed powers are obtained for the \gls{iid} Rayleigh fading channel given in $\mathbf{(As3)}$.

\subsection{\glspl{pa} Consumed Power}
The following theorem gives the asymptotic expression of the \glspl{pa} consumed power by the precoder (\ref{eq:ZF_pcons_sol}). This characterization is possible because, in the asymptotic wideband regime, all the \gls{bs} antennas are approximately allocated the same power. 
\begin{theorem}
\label{th:p_PAs_asympt}
Under $\mathbf{(As1)}-\mathbf{(As3)}$, the transmit power allocated to each antenna by the precoder~(\ref{eq:ZF_pcons_sol}) converges to
\begin{equation}
\label{eq:per-antenna_ptx}
p_m \to \overline{p} = \frac{1}{M(M-K)}\mathrm{tr}\left(\mat{D}_{\vect{\beta}}^{-1}\mat{D}_{\vect{\gamma}}\sigma_\nu^2\right)
\end{equation}
implying that the total power consumed by the \glspl{pa} satisfies
\begin{equation}
\label{eq:p_cons_asympt}
    p_\mathrm{PAs} \to \overline{p_\mathrm{PAs}} = \alpha\left(\frac{M}{M-K}\mathrm{tr}\left(\mat{D}_{\vect{\beta}}^{-1}\mat{D}_{\vect{\gamma}}\sigma_\nu^2\right)\right)^{1/2}.
\end{equation}
The per-subcarrier \gls{zf} precoder is therefore found back
\begin{equation}
    \mat{W}_q = \hr{\mat{H}_q}\left(\mat{H}_q\hr{\mat{H}_q}\right)^{-1}\tilde{\mat{D}}_{\vect{\gamma}}^{1/2}\sigma_\nu.
\end{equation}
\begin{proof}
See Appendix~\ref{proof2}.
\end{proof}
\end{theorem}
\noindent Looking at (\ref{eq:p_cons_asympt}), the power consumed by the \glspl{pa} is a monotonically decreasing function of $M$: activating more antennas is then always beneficial. However, the marginal gain of activating more antennas decreases as $M$ increases since $\frac{M}{M-K}$ converges to a unit constant.

\subsection{\gls{bs} Consumed Power}

Considering the power consumption of the \gls{bs}, the optimal precoder does not necessarily activate all the antennas. Indeed, the power consumed by the \gls{bs} circuits induces a penalty on the number of active antennas $M_\mathrm{a}$. The general problem, considering the \gls{bs} consumption and the per-antenna power constraints, consists in solving
\begin{equation}
\label{eq:p_cons,BS_papc_gen}
\begin{aligned}
    & \minimize_{\substack{\{w_{m,k,q}\},M_\mathrm{a}\\m=0,\dotsc,M_\mathrm{a}-1}}  \quad  	&& 	p_\mathrm{BS} = \alpha\sum_{m=0}^{M_\mathrm{a}-1}p_m^{1/2}+p_{\mathrm{fix}}+\mathcal{C}M_\mathrm{a} \\
    & \mathrm{\ subject\ to} \quad		    && 	\mat{H}_q\mat{W}_q = \tilde{\mat{D}}_{\vect{\gamma}}^{1/2}\sigma_\nu \;\; \forall q, \\
    &                                       &&   p_m \leq p_{\mathrm{max}} \;\; \forall m, \\
    &                                       &&   M_\mathrm{a} \leq M. \\
\end{aligned}
\end{equation}
Without loss of generality, one can first optimize with respect to the precoding coefficients considering a generic number of active antennas $M_\mathrm{a}$, and then find the optimal integer number of active antennas.

Let us consider the power allocation defined in Theorem~\ref{th:p_PAs_asympt} among the $M_\mathrm{a}$ active antennas. By combining the \gls{bs} consumption model in~(\ref{eq:total_consumed_power_BS}) with the asymptotic \glspl{pa} consumption in~(\ref{eq:p_cons_asympt}), the asymptotic \gls{bs} consumption is computed as
\begin{equation}
\label{eq:p_cons,BS_asympt}
p_\mathrm{BS} \to \overline{p_\mathrm{BS}} = \alpha\left(\frac{M_\mathrm{a}}{M_\mathrm{a}-K}\mathrm{tr}\left(\mat{D}_{\vect{\beta}}^{-1}\mat{D}_{\vect{\gamma}}\sigma_\nu^2\right)\right)^{1/2} + p_{\mathrm{fix}}+\mathcal{C}M_\mathrm{a}.
\end{equation}
The above quantity is the solution, under $\mathbf{(As1)-(As3)}$, to the minimization problem~(\ref{eq:p_cons,BS_papc_gen}) with respect to the precoding coefficients and considering a fixed number of active antennas $M_\mathrm{a}$.
The corresponding precoder is a per-subcarrier \gls{zf} precoder among the $M_\mathrm{a}$ active antennas.

We can now optimize the number of active antennas $M_\mathrm{a}$.

\subsubsection{Optimal $M_\mathrm{a}$ without Power Constraints}
In this case, we can directly minimize~(\ref{eq:p_cons,BS_asympt}) and check that the solution is within the allowed bounds.
\begin{lemma}
\label{lemma1}
Under $\mathbf{(As1)}-\mathbf{(As3)}$, the optimal integer number of active antennas $M^\star_\mathrm{a}\in\mathbb{N}$, $K+1\leq M^\star_\mathrm{a}\leq M$, minimizing $\overline{p_\mathrm{BS}}$ is given by
\begin{equation}
\label{eq:M_a-star}
M_\mathrm{a}^\star = \left\lfloor \min{\left\{\max{\left\{\tilde{M}_\mathrm{a},K+1\right\}},M\right\}} \right\rceil
\end{equation}
where $\tilde{M}_\mathrm{a}=x$, $x$ being the solution to the quartic equation
\begin{equation}
\label{eq:quartic_eq}
x(x-K)^3-\frac{\alpha\left(\tr{\mat{D}_{\vect{\beta}}^{-1}\mat{D}_{\vect{\gamma}}\sigma_\nu^2}\right)^{1/2}K}{2\mathcal{C}}=0
\end{equation}
for $x>K$. Closed-form solutions to (\ref{eq:quartic_eq}) can be found via~\cite{Shmakov11}.
\begin{proof}
See Appendix~\ref{proof3}.
\end{proof}
\end{lemma}
Two opposite cases can occur: $M_\mathrm{a}^\star=K+1$ (circuits power-limited regime) and $M_\mathrm{a}^\star=M$ (\glspl{pa} power-limited regime). Note that the ceil-floor operator $\left\lfloor\cdot\right\rceil$ is evaluated after having checked the bounds. Indeed, if $M_\mathrm{a}^\star$ is equal to one of the two extremes, the ceil-floor operator becomes trivial.

\subsubsection{Optimal $M_\mathrm{a}$ with Power Constraints}
Up to now, there is no guarantee that the antenna powers by using $M_\mathrm{a}^\star$ antennas satisfy the per-antenna power constraints. One has to solve problem~(\ref{eq:p_cons,BS_papc_gen}) by relaxing $\mathbf{(As1)}$. In Appendix~\ref{proof4} we show how, under $\mathbf{(As2)}$ and $\mathbf{(As3)}$, the per-antenna power constraints reduce to $\overline{p}\leq p_\mathrm{max}$, where $\overline{p}$ is the asymptotic transmit power at each active antenna given in~(\ref{eq:per-antenna_ptx}). It is therefore sufficient to compute the continuous number of antennas that gives exactly $p_\mathrm{max}$ as transmit power, and then apply the ceiling operator to prevent the constraint from being violated. This value is then included as a lower bound on the final solution.

In the following, we assume that problem (\ref{eq:p_cons,BS_papc_gen}) is feasible, i.e., activating all antennas is sufficient to meet the \gls{qos} constraints. This is expressed, using~(\ref{eq:per-antenna_ptx}), by the following condition:

$(\mathbf{As4)}$: $\frac{1}{M(M-K)}\mathrm{tr}\left(\mat{D}_{\vect{\beta}}^{-1}\mat{D}_{\vect{\gamma}}\sigma_\nu^2\right) \leq p_\mathrm{max}$, 
\\which ensures that activating $M_\mathrm{a}=M$ antennas does not violate the per-antenna power constraints.

\begin{theorem}
\label{th:Ma_dag}
Under $\mathbf{(As2)-(As4)}$, the optimal integer number of active antennas solving problem (\ref{eq:p_cons,BS_papc_gen}) is given by
\begin{equation}
\label{eq:Mdag}
M_\mathrm{a}^\dag = \left\lfloor \min{\left\{\max{\left\{\tilde{M}_\mathrm{a},K+1,\hat{M}_\mathrm{a}\right\}},M\right\}} \right\rceil
\end{equation}
where $\tilde{M}_\mathrm{a} =  x $, $x$ being the solution to (\ref{eq:quartic_eq}), and
\begin{equation}
\label{eq:Mhat}
    \hat{M}_\mathrm{a} = 
    \left\lceil\frac{1}{2}\left(K+\left(K^2+\frac{4\tr{\mat{D}_{\vect{\beta}}^{-1}\mat{D}_{\vect{\gamma}}\sigma_\nu^2}}{p_\mathrm{max}}\right)^{1/2}\right)\right\rceil.
\end{equation}
The quantity $\hat{M}_\mathrm{a}$ is the minimal integer number of active antennas required to satisfy the per-antenna power constraint and achieve the users' \gls{qos}.
\end{theorem}
\begin{proof}
See Appendix~\ref{proof4}.
\end{proof}
We stress that all the involved quantities depend only on the large-scale fading coefficients, i.e., the second-order statistics of the channel. Therefore, instantaneous \gls{csi} is not required and $M_\mathrm{a}^\dag$ does not need to be computed every channel coherence time. As in Lemma~\ref{lemma1}, the ceil-floor operator is evaluated at last to avoid unnecessary operations. Algorithm~\ref{alg:Madag} illustrates the different steps that are performed to assign the value of $M_\mathrm{a}^\dag$.

\begin{algorithm}[t!]\linespread{1}
	\caption{Find optimal number of active antennas in asymptotic wideband system}\label{alg:Madag}
	\small 
	\begin{algorithmic}
		\Require $\gamma_0,\dotsc,\gamma_{K-1},\beta_0,\dotsc,\beta_{K-1}$ \Comment{Set large-scale fading-based parameters}
			\State $\hat{M}_\mathrm{a} \gets (\ref{eq:Mhat})$ \Comment{Min. number of \gls{bs} antennas required by per-antenna power constraints}
			\State $\tilde{M}_\mathrm{a} \gets \mathrm{Sol.\ of\ } (\ref{eq:quartic_eq})$ \Comment{Number of \gls{bs} antennas in $[K,\infty[$ minimizing asymptotic \gls{bs} consumption}
			\State $y\gets \max\left\{\hat{M}_\mathrm{a},\tilde{M}_\mathrm{a}\right\}$
			\If{$y \leq K+1$}
				\State $M_\mathrm{a}^\dag \gets K+1$ \Comment{Use as less antennas as possible}
			\ElsIf{$y \geq M$}
				\State $M_\mathrm{a}^\dag \gets M$ \Comment{Use as many antennas as possible}
			\Else
				\State $M_\mathrm{a}^\dag \gets \left\lfloor y \right\rceil$ \Comment{Compute ceil-floor operation on intermediate number of antennas}
			\EndIf
	\end{algorithmic}
\end{algorithm}

%% file: 6_Simulation_Results.tex
\section{Performance Evaluation}
\label{sec:sim_results}

In the numerical experiments, the large-scale fading coefficient (in dB) for the user $k$ is computed as~\cite{Bjornson15}
\begin{equation}
    \beta_k^{\mathrm{dB}} = -35.3-37.6\log_{10}u_k
\end{equation}
where $u_k$ is the distance in meters between the user $k$ and the \gls{bs}. We consider the users to be uniformly distributed within a circular cell delimited by $u_\mathrm{min}$ and $u_\mathrm{max}$. The cumulative density function of $u_k$ is given by the ratio between the area of the annulus bounded by $u_k$ and the area of the largest possible annulus, i.e., $F_{u_k}(\upsilon)=\frac{\upsilon^2-u_\mathrm{min}^2}{u_\mathrm{max}^2-u_\mathrm{min}^2}$. Therefore, the probability density function of $u_k$ is $f_{u_k}(\upsilon) = \frac{dF_{u_k}(\upsilon)}{d\upsilon} = \frac{2\upsilon}{u_\mathrm{max}^2-u_\mathrm{min}^2}$. The target \gls{sinr} (in dB) of the user $k$ is then computed as
\begin{equation}
\label{eq:snr_per-user}
    \gamma_k^{\mathrm{dB}} = 5\log_{10}\left(\frac{\beta_k}{4.86\times 10^{-14}}\right).
\end{equation}
Using (\ref{eq:snr_per-user}), $\gamma_k^{\mathrm{dB}} \in [4, 20]$ dB.
The channel model is given in $\mathbf{(As3)}$. The remaining parameters are listed in Table~\ref{tab:sim_params}. As performance metrics, we consider the gain in $p_\mathrm{PAs}$ and the gain in $p_\mathrm{BS}$. The gain in $p_\mathrm{PAs}$ is the ratio between the \glspl{pa} consumed power by the transmit power-based precoder and by the \glspl{pa} consumed power-based precoder. On the other hand, the gain in $p_\mathrm{BS}$ is the ratio between the \gls{bs} consumed power by the transmit power-based precoder and by the \glspl{pa} consumed power-based precoder. In this way, we are using as a benchmark the conventional \gls{zf} precoder, which optimizes the transmit power and activates all \gls{bs} antennas.
\begin{table}[htpb]
\centering
\caption{Parameters of the numerical experiments.}
\label{tab:sim_params}
\begin{tabular}{c|c}
    \hline
    \textbf{Parameter} & \textbf{Value} \\
    \hline
    \glspl{pa} maximal power: $p_\mathrm{max}$ & $1$ W \\ 
    \glspl{pa} maximal efficiency: $\eta_\mathrm{max}$ & $0.22$ \\
    Noise power: $\sigma_\nu^2$ & $-96$ dBm \\
    Fixed power consumption: $p_{\mathrm{fix}}$ & $15$ W \\ 
    Circuit power consumption scaling per antenna: $\mathcal{C}$ & $0.7$ W \\
    Minimum distance from user to BS: $u_\mathrm{min}$ & $35$ m \\ 
    Maximum distance from user to BS: $u_\mathrm{max}$ & $250$ m \\ 
    \hline
\end{tabular}
\end{table}

\textit{Remark: the simulations of the narrowband and wideband systems analyzed in Sections~\ref{sec:SC_system} and~\ref{sec:MC_system} assume no maximal per-antenna power constraints, as the precoders are obtained under $p_\mathrm{max}\to\infty$. However, $p_\mathrm{max}$ has to be used to quantify the consumed powers. Therefore, in the simulations, the results exceeding the $p_\mathrm{max}$ values are discarded.}

When assessing the performance of asymptotic wideband systems, we run simulations for $M=64$ varying the value of $K$ from $1$ (low traffic) to $40$ (high traffic). Using the parameters in Table~\ref{tab:sim_params} and considering the conventional case $M_\mathrm{a}=M$, the average contributions of the terms $p_\mathrm{PAs}$, $\mathcal{C}M_\mathrm{a}$, and $p_\mathrm{fix}$ are $7\%$, $70\%$, and $23\%$ for $K=1$, then $31\%$, $52\%$, and $17\%$ for $K=20$, while they are $46\%$, $40\%$, and $14\%$ for $K=40$.

\subsection{Precoding Solutions for a Channel Realization}

\begin{figure*}[!ht]
\centering
\subfloat[]{
\resizebox{!}{5.5cm}{
\includegraphics{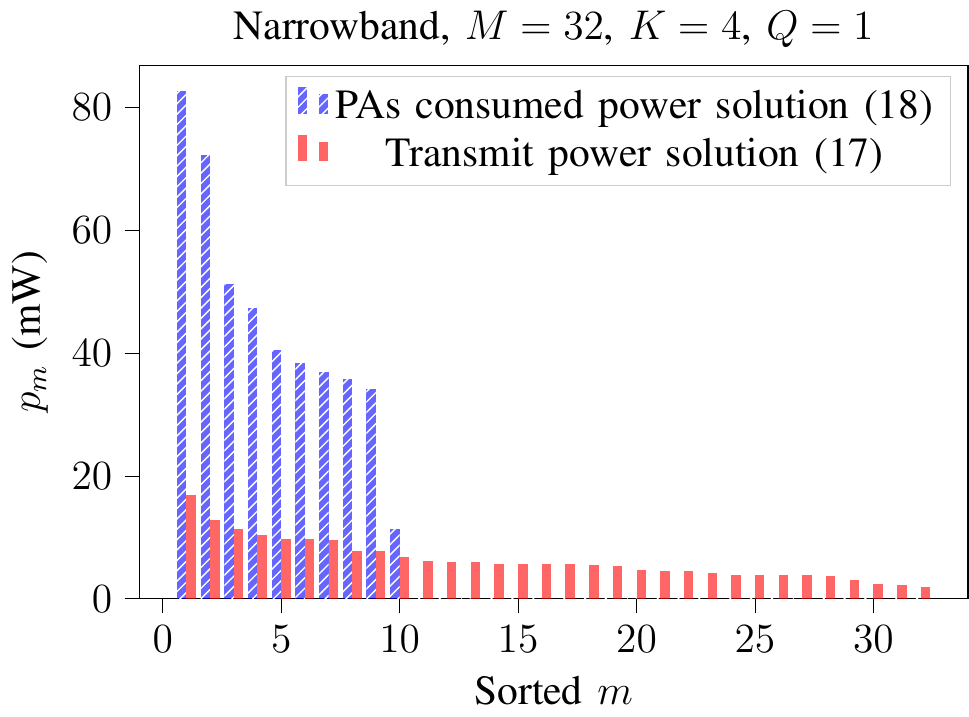}
}
\label{fig:pm-vs-m_SC}}
\hfil
\subfloat[]{
\resizebox{!}{5.5cm}{
\includegraphics{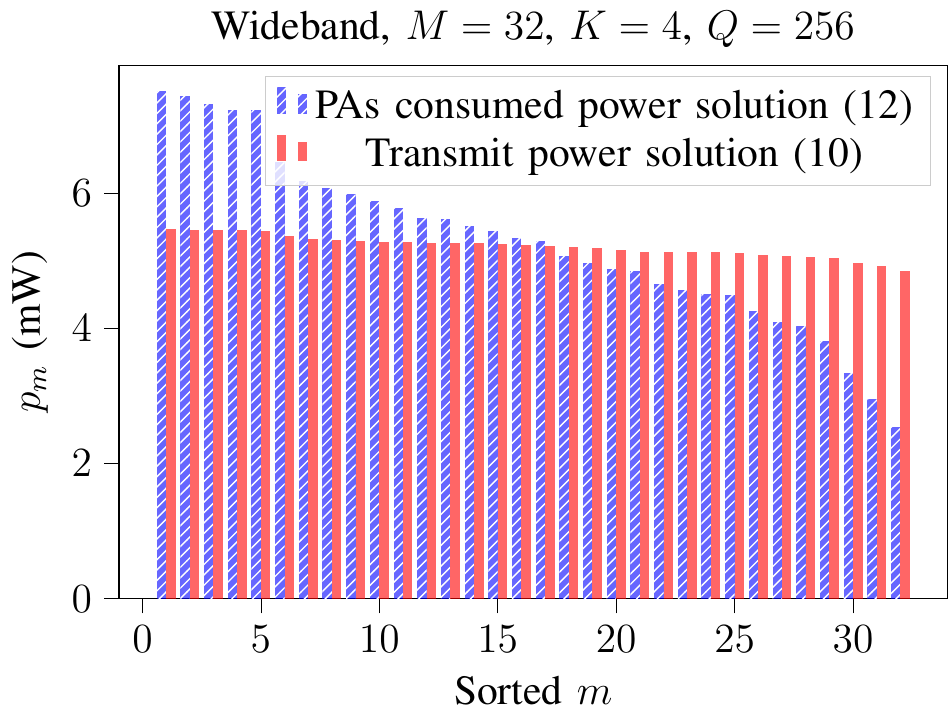}
}
\label{fig:pm-vs-m_MC}}
\caption{Transmit power as a function of the antenna index with decreasing allocated power, for $M=32$ antennas and $K=4$ users in \gls{iid} Rayleigh fading. (a) Narrowband case, $Q=1$ subcarrier. The gain in \glspl{pa} consumption is $1.19$, while the gain in \gls{bs} consumption is $1.51$. (b) Wideband case, $Q=256$ subcarriers. The gains in \glspl{pa} and \gls{bs} consumption are both equal to $1.00$.}
\end{figure*}

Fig.~\ref{fig:pm-vs-m_SC} shows the solutions, in terms of transmit powers at the different antennas, for the narrowband system and a specific channel realization, $K=4$ and $M=32$. As discussed in~\cite{Peschiera22}, the solution based on the \glspl{pa} consumed power activates only few antennas, while the traditional \gls{zf} solution uses all the available antennas. Using (\ref{eq:total_consumed_power}) as a cost function induces sparsity in the power allocation.
Note that, differently from the single-user case where only one antenna is used, the number of active antennas must be greater than $K$ to enable spatial multiplexing. 

When moving to a wideband system, the precoder based on the \glspl{pa} consumed power gradually uses more antennas as $Q$ increases. Fig.~\ref{fig:pm-vs-m_MC} shows the solutions for $Q=256$. The solution based on the \glspl{pa} consumed power, although presenting more variability in the powers per antenna than the traditional solution, still activates all the \gls{bs} antennas. Achieving the \gls{sinr} constraints for all the subcarriers, while minimizing the consumed power by the \glspl{pa}, requires using all the available antennas. 
In short, Fig.~\ref{fig:pm-vs-m_MC} shows that the per-subcarrier \gls{zf} (\ref{eq:ZF_ptx_sol}) is efficient in terms of \glspl{pa} consumed power.

\subsection{Average Gains in Power Consumption}

\begin{figure*}[!ht]
\centering
\subfloat[]{
\resizebox{!}{6cm}{
\includegraphics{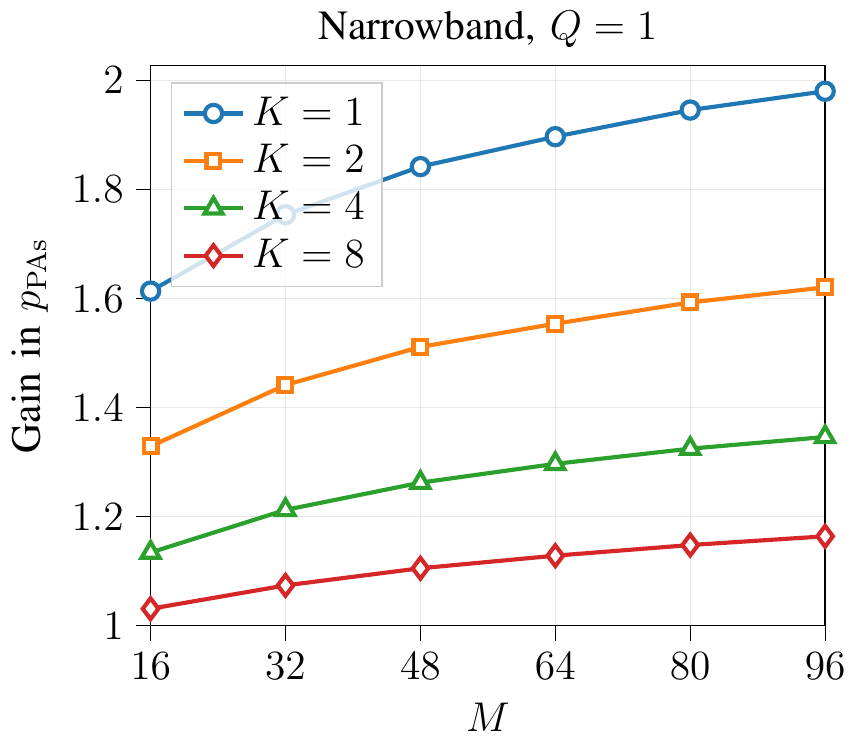}
}
\label{fig:gain_pPAs_SC}}
\hfil
\subfloat[]{
\resizebox{!}{6cm}{
\includegraphics{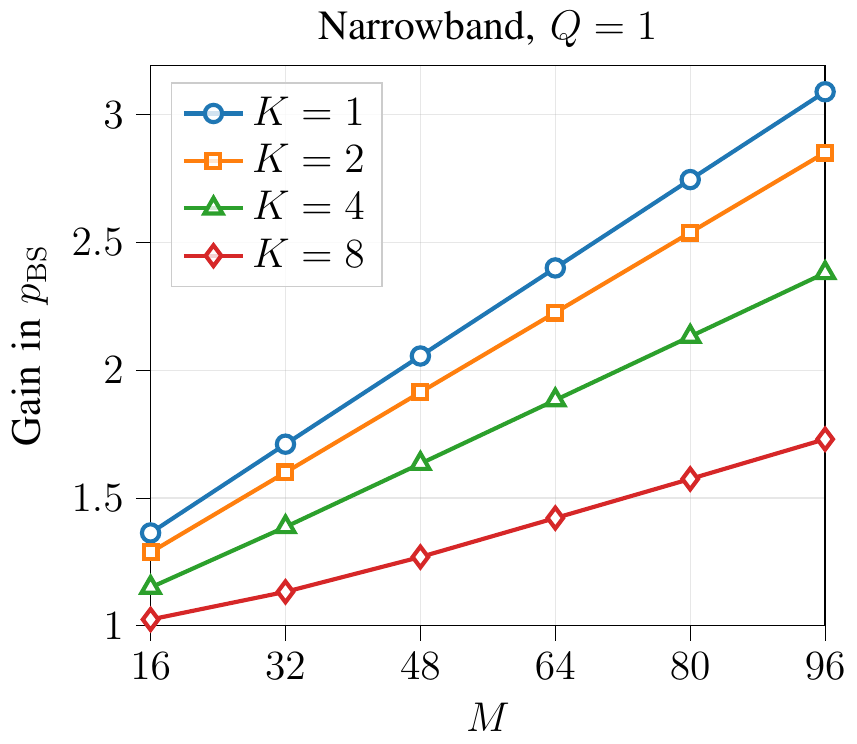}
}
\label{fig:gain_pBS_SC}}
\hfil
\subfloat[]{
\resizebox{!}{6cm}{
\includegraphics{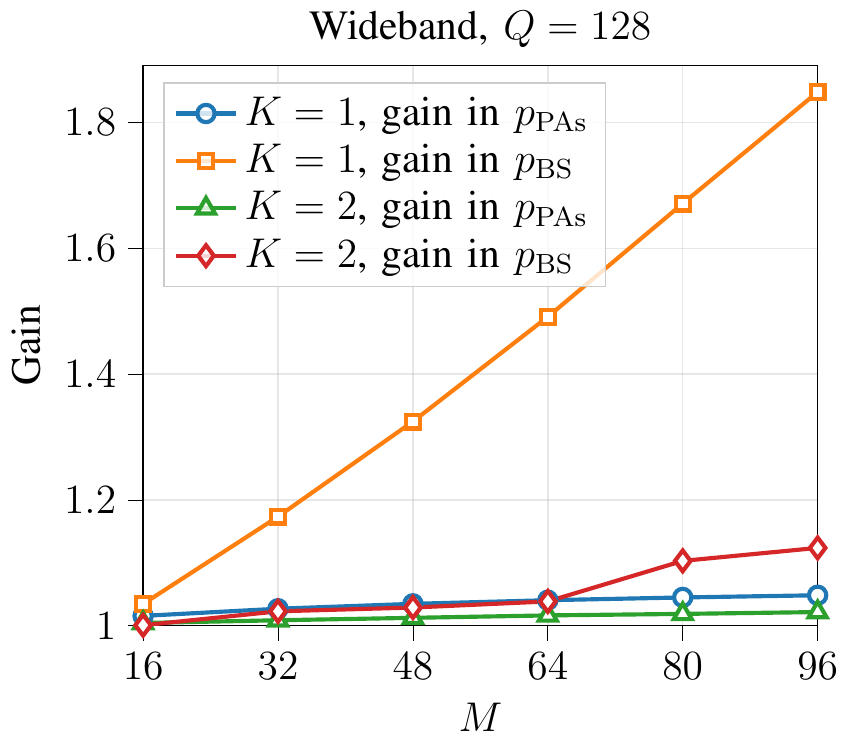}
}
\label{fig:gain_MC}}
\hfil
\subfloat[]{
\resizebox{!}{6cm}{
\includegraphics{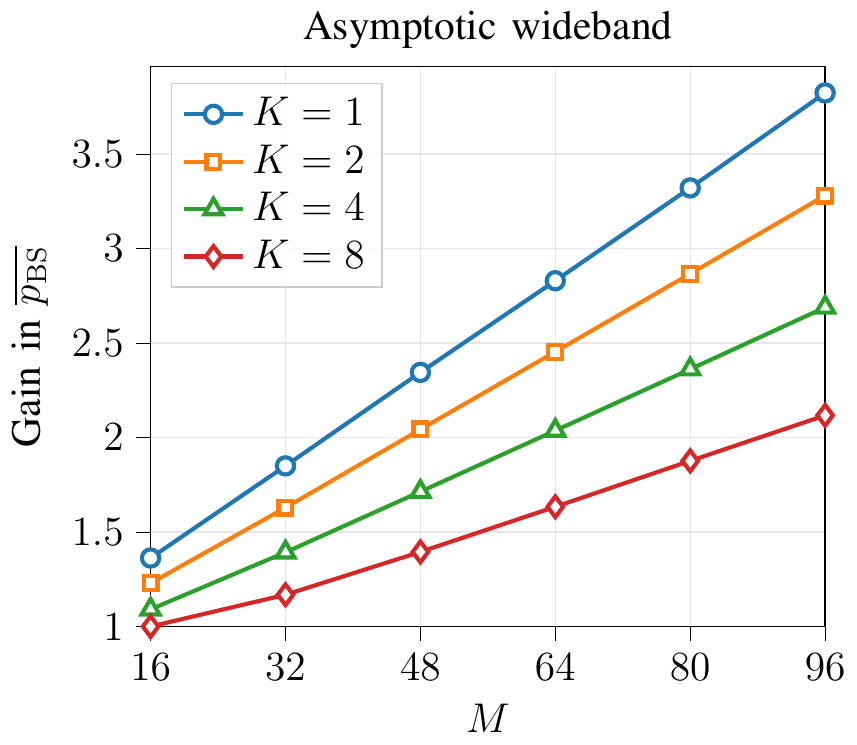}
}
\label{fig:gain_asMC}}
\caption{Gain in power consumption versus number of antennas for narrowband and wideband systems in \gls{iid} Rayleigh fading ($2\cdot10^3$ realizations), shown for different values of $K$. (a) Narrowband, gain in \glspl{pa} consumption. (b) Narrowband, gain in \gls{bs} consumption. (c) Wideband with $Q=128$, gain in \glspl{pa} and \gls{bs} consumption. (d) Asymptotic wideband, gain in asymptotic \gls{bs} consumption. In this case, the gain is the \gls{bs} consumption by activating $M$ antennas divided by the \gls{bs} consumption by activating $M_\mathrm{a}^\dag$ antennas.}
\label{fig:gains}
\end{figure*}

Following the previous observations, Fig.~\ref{fig:gain_pPAs_SC}--\ref{fig:gain_asMC} show the relative differences in the consumed power for different systems. In the narrowband system and considering the \glspl{pa} consumed power to evaluate the performance (Fig.~\ref{fig:gain_pPAs_SC}), the gain of the novel precoder (\ref{eq:ZF_pcons_sol_SC}) over the traditional one (\ref{eq:ZF_ptx_sol_SC}) is always greater than $1.5$ for $K=1$. When more users are present, the difference significantly decreases. For $K=8$, the ratio is below $1.2$ for every value of $M$. Instead, when the \gls{bs} consumed power is evaluated, the achievable gains remain large even when more users are communicating (Fig.~\ref{fig:gain_pBS_SC}). With $M=64$, the gain in the \gls{bs} power consumption ranges from $1.4$ to $2.4$, depending on $K$.

In a wideband system with $Q=128$ subcarriers, the differences between the conventional and the novel precoder remain significant for $K=1$ and considering the \gls{bs} consumption (Fig.~\ref{fig:gain_MC}). However, the presence of $K=2$ users already requires the activation of all the antennas, and both the \glspl{pa} and \gls{bs} power consumption gains tend to $1$.

Considering the asymptotic wideband case, one can compute the gains achieved by using $M_\mathrm{a}^\dag$ antennas over $M$. Fig.~\ref{fig:gain_asMC} shows that significant savings in \gls{bs} consumption can be obtained, up to a factor of $3.8\times$. The benefits are similar to the ones achieved in a narrowband system when evaluating the \gls{bs} consumption. For instance, when $K=4$ and $M=48$, the conventional precoder consumes $1.7$ times more power than the optimized precoder.

\subsection{Asymptotic Wideband Regime}

\begin{figure*}[!b]
\centering
\subfloat[]{
\resizebox{!}{6cm}{
\includegraphics{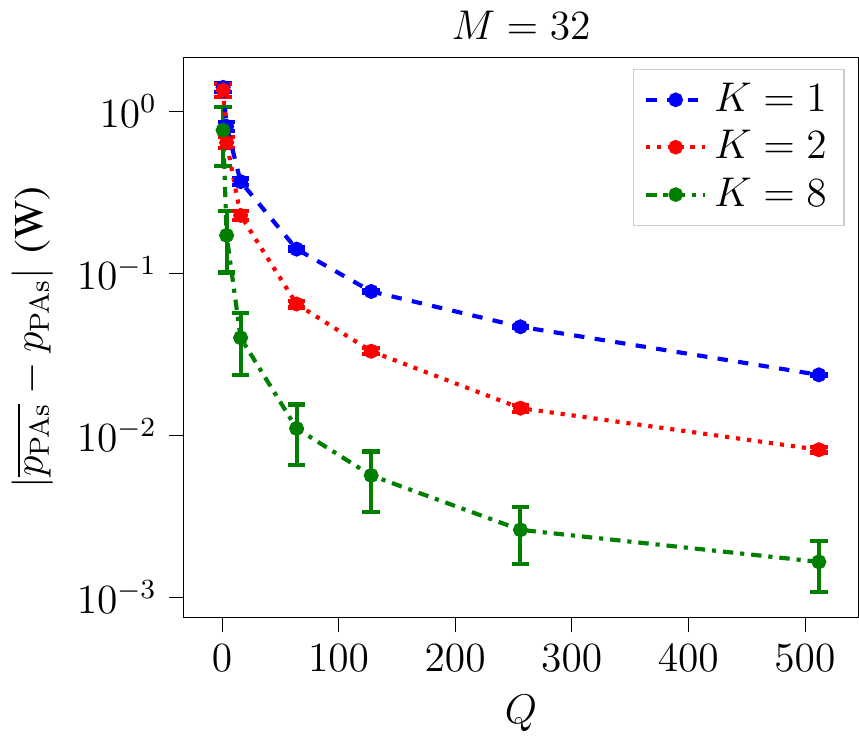}
}
\label{fig:as_error}}
\hfil
\subfloat[]{
\resizebox{!}{6cm}{
\includegraphics{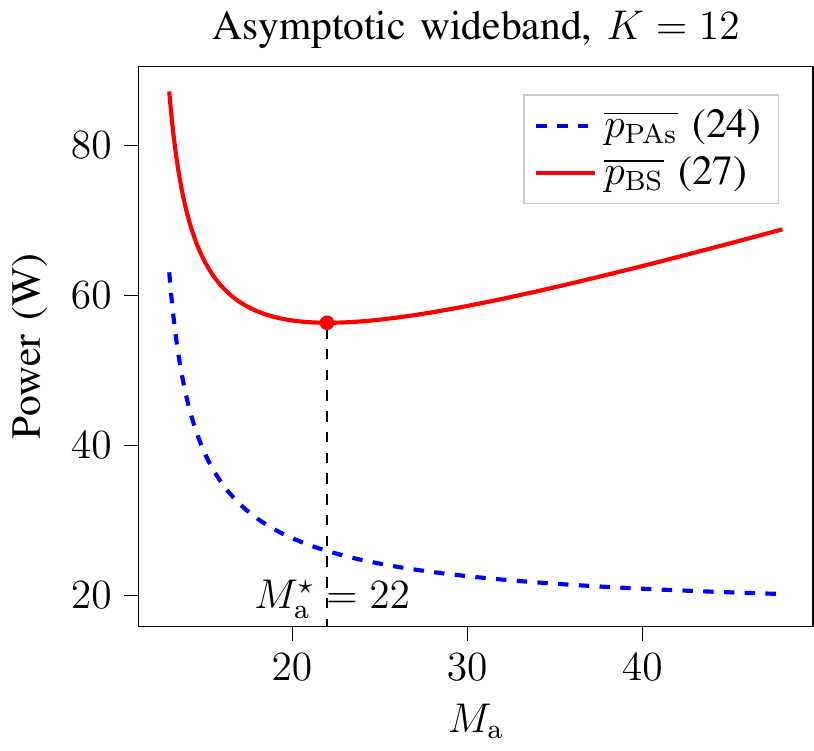}
}
\label{fig:as_example}}
\caption{(a) Magnitude of the difference between the asymptotic and non-asymptotic \glspl{pa} consumptions versus number of subcarriers for $M=32$ and different values of $K$. The non-asymptotic consumptions are computed through simulations, and the confidence intervals of the results are depicted, where the width is set to two times the variance of the results. The differences become smaller as $Q$ increases. (b) Asymptotic \glspl{pa} and \gls{bs} power consumption versus number of active antennas for $M=48$ and $K=12$. The number of antennas $M_\mathrm{a}^\star$ is also shown.}
\end{figure*}

Fig.~\ref{fig:as_error} illustrates the validity of the asymptotic results in Section~\ref{sec:asymptotic_MC} for $M=32$ and different values of $K$. The simulated \glspl{pa} consumed powers converge to the asymptotic values, and $Q=128$ is already sufficient to allow the asymptotic result to be a tight approximation (i.e., the average errors are below $10^{-1}$). The trend of the asymptotic \glspl{pa} and \gls{bs} consumed powers as a function of $M_\mathrm{a}$ is shown in Fig.~\ref{fig:as_example}, for $K=12$. The \glspl{pa} consumed power decreases monotonically with $M_\mathrm{a}$, while the \gls{bs} consumed power first decreases for small $M_\mathrm{a}$ (it is beneficial to use more antennas to reduce the \glspl{pa} consumed power) and then increases for large $M_\mathrm{a}$ (it is detrimental to use more antennas due to the larger circuit power consumption). The convexity of (\ref{eq:p_cons,BS_asympt}) implies a global minimum. 

The achievable power consumption savings by activating $M_\mathrm{a}^\dag$ antennas can be also characterized as a function of the number of users. Fig.~\ref{fig:asympt_K} illustrates the comparison for $M=64$. We consider the number of users to be an indication of the system traffic and load. Indeed, the per-antenna power constraints are not binding when $K$ is small, and they gradually start to be active when $K$ increases. The conventional precoder always utilizes all the $M$ antennas, while the optimized precoder adapts the number of utilized antennas to the traffic. The value of $M_\mathrm{a}^\dag$ grows proportionally to $K$, as shown in Fig.~\ref{fig:asK2}.\footnote{The same trend is observed in~\cite{Senel19}, with the difference that in~\cite{Senel19} the authors considered $p_\mathrm{BS}$ to be the sum of the transmit power and the circuit power consumption.} When the system becomes more loaded, $M_\mathrm{a}^\dag$ increases due to the larger value of $p_\mathrm{PAs}$ and due to the transmit power constraint (i.e., more antennas need to be saturated to achieve the \gls{qos} constraints). The \gls{bs} consumptions by activating $M$ and $M_\mathrm{a}^\dag$ antennas eventually converge (Fig.~\ref{fig:asK1}). We also make a comparison with a precoder that always uses the minimum number of antennas ($K+1$). This precoder minimizes the \gls{bs} consumption when $K=1$, and then starts to diverge from the optimal precoder because it does not counterbalance the increase in $p_\mathrm{PAs}$ with the utilization of more antennas. 

In terms of achievable savings (i.e., the ratio between the asymptotic \gls{bs} consumptions of the benchmark and optimized precoders), Fig.~\ref{fig:asK2} shows that the optimized precoder reduces up to a factor of $2.8\times$ the consumption in low load with respect to the conventional precoder. When $K=10$, we can still achieve a gain equal to $1.5$. For more users, the gain reduces and eventually vanishes (all antennas are activated in both cases). The savings over the precoder that activates $K+1$ antennas follow the opposite trend, increasing in high load up to a factor of $2.2\times$. 

\begin{figure*}[!t]
\centering
\subfloat[]{
\resizebox{!}{6.2cm}{
\includegraphics{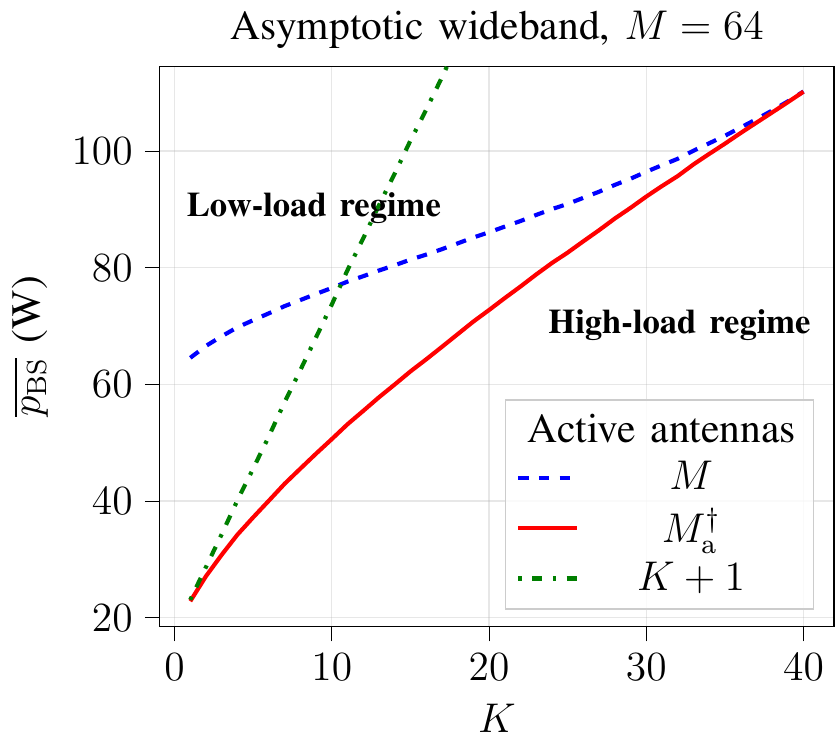}
}
\label{fig:asK1}
}
\hfil
\subfloat[]{
\resizebox{!}{6.2cm}{
\includegraphics{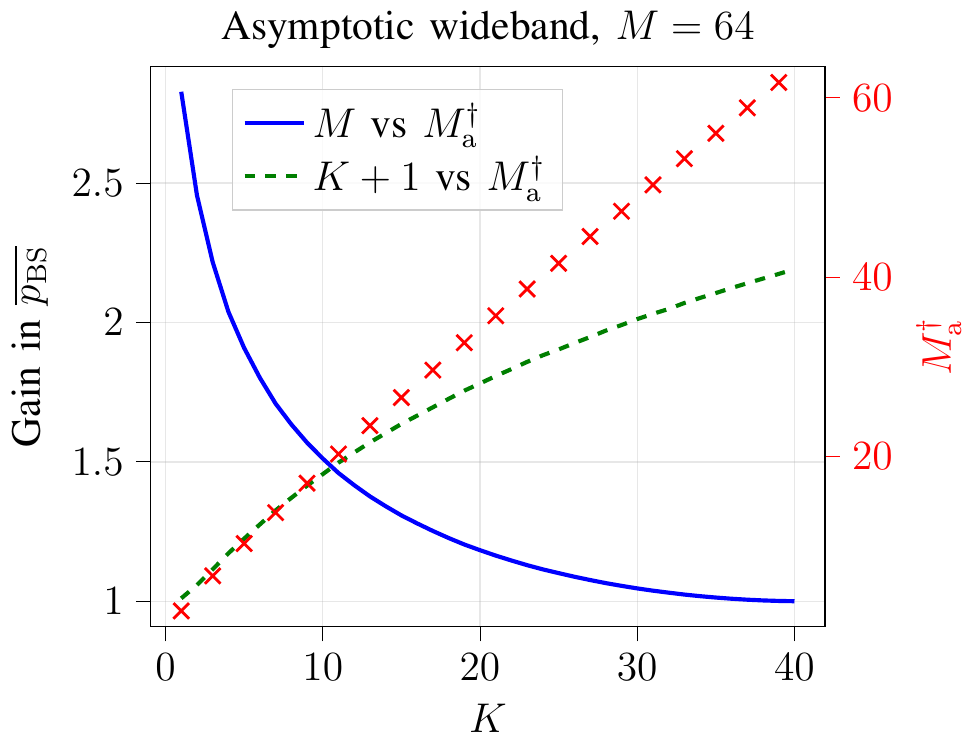}
}
\label{fig:asK2}
}
\caption{Characterization of the asymptotic wideband system as a function of $K$ for $M=64$, averaged over $2\cdot 10^3$ channel realizations. (a) Asymptotic \gls{bs} power consumption (\ref{eq:p_cons,BS_asympt}) by using $M$, $K+1$, and $M_\mathrm{a}^\dag$ antennas. (b) Referring to the same simulations, ratio between $\overline{p_\mathrm{BS}}$ by activating $M$ antennas and $\overline{p_\mathrm{BS}}$ by activating $M_\mathrm{a}^\dag$ antennas (solid line), ratio between $\overline{p_\mathrm{BS}}$ by activating $K+1$ antennas and $\overline{p_\mathrm{BS}}$ by activating $M_\mathrm{a}^\dag$ antennas (dashed line), and optimized number of active antennas $M_\mathrm{a}^\dag$ (crosses).}
\label{fig:asympt_K}
\end{figure*}

\subsection{Complexity Analysis}

To quantify the computational complexity of the proposed precoders, we calculate the number of complex floating point operations (flops). As a reference for the number of flops required by standard linear algebraic operations, we consider~\cite[App. C.1.1]{Boyd04}. 

\subsubsection{Wideband System} In the wideband system, we need to calculate~(\ref{eq:ZF_pcons_sol}). The computation of both $\mat{H}_q\mat{D}_{\vect{p}}^{1/2}$ and $\mat{D}_{\vect{p}}^{1/2}\hr{\mat{H}_q}$ require $KM$ flops, corresponding to a scalar multiplication for each element of $\mat{H}_q$. Computing $\mat{H}_q\mat{D}_{\vect{p}}^{1/2}\hr{\mat{H}_q}$, which is a symmetric matrix, requires $K^2M$ flops. Using Cholesky factorization and forward and backward substitution, we can compute the $m$-th row of the matrix $\mat{D}_{\vect{p}}^{1/2}\hr{\mat{H}_q}\inv{\mat{H}_q\mat{D}_{\vect{p}}^{1/2}\hr{\mat{H}_q}}$ with $\frac{1}{3}K^3+2K^2$ flops. The forward and backward substitutions need to be performed $M$ times, therefore the number of flops becomes $\frac{1}{3}K^3+2K^2M$. At the end, the scaling by $\tilde{\mat{D}}_{\vect{\gamma}}^{1/2}\sigma_\nu$ requires $K$ flops. All that needs to be repeated for every subcarrier. By adding the terms, we obtain $\frac{1}{3}K^3Q+3K^2MQ+2KMQ+KQ$ flops, which scales as $\mathcal{O}\left(K^3Q+K^2MQ\right)$. Before computing~(\ref{eq:ZF_pcons_sol}), the powers at the antennas need to be retrieved. The computation of~(\ref{eq:MC-MU_fpe}) requires, other than the necessary flops to compute $\mat{W}_q$, $KQ$ multiplications to calculate the magnitudes and $K+Q-2$ additions to calculate the sums. Multiplying by $M$, we obtain $KQM+KM+QM-2M$ additional flops. The trend, however, remains $\mathcal{O}\left(K^3Q+K^2MQ\right)$. By using Algorithm~\ref{alg:fpim}, the computation of powers is repeated until convergence. If the algorithm converges in $I$ iterations, the total number of flops is $\mathcal{O}\left(K^3QI+K^2MQI\right)$. The conventional \gls{zf}, instead, requires $\mathcal{O}\left(K^3Q+K^2MQ\right)$ flops.

\begin{figure}[!ht]
\centering
\resizebox{!}{6.5cm}{
\includegraphics{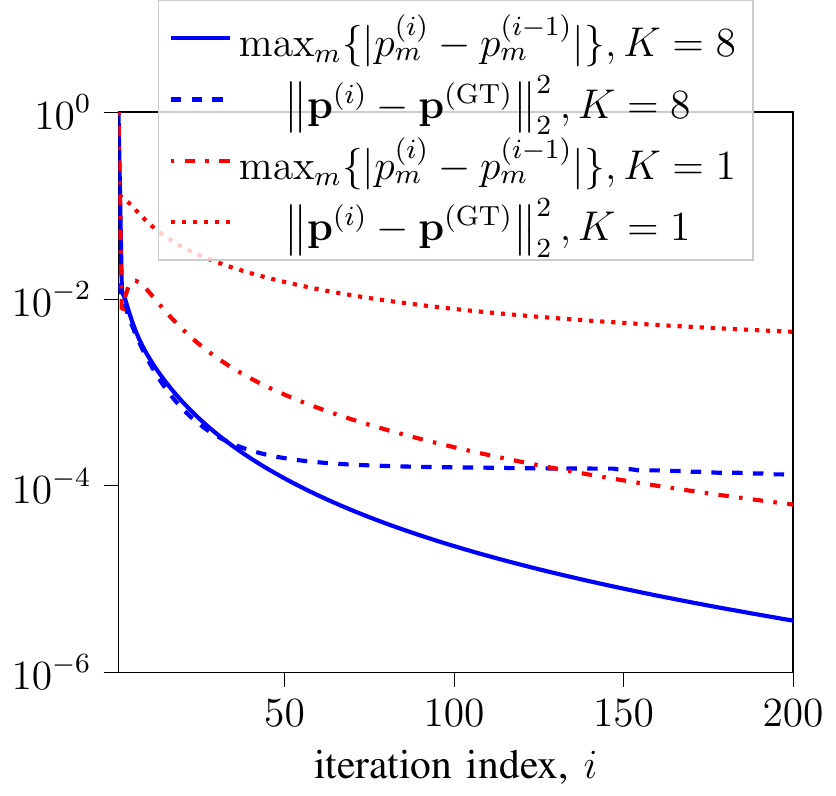}
}
\caption{Performance of Algorithm~\ref{alg:fpim} when $Q=1$ and $M=32$, using a tolerance $\varepsilon=10^{-4}$. The ground truth $\vect{p}^\mathrm{(GT)}$ is computed with a convex solver. Results are averaged over $10^3$ executions of the algorithm.}
\label{fig:fpim}
\end{figure}

\begin{table}[t!]
\centering
\caption{Precoders complexity in the three analyzed systems.}
\label{tab:compl}
\begin{tabular}{c|c|c}
    \hline
    \textbf{System} & \textbf{Proposed sol.} & \textbf{Conventional sol.} \\
    \hline
    Wideband & $\mathcal{O}\left(K^3QI+K^2MQI\right)$ & $\mathcal{O}\left(K^3Q+K^2MQ\right)$ \\
    Narrowband & $\mathcal{O}\left(K^3I+K^2MI\right)$ & $\mathcal{O}\left(K^3+K^2M\right)$ \\
    As.\ wideband & $\mathcal{O}\left(K^3Q+K^2M_\mathrm{a}Q\right)$ & $\mathcal{O}\left(K^3Q+K^2MQ\right)$ \\
    \hline
\end{tabular}
\end{table}

Fig.~\ref{fig:fpim} illustrates the convergence of the iterative algorithm. We show two narrowband cases because, when $Q$ increases, the convergence is faster (i.e., in tens of iterations), likely due to the more uniform power allocation that might be associated with fewer changes from one iteration to the other. We display the maximum absolute variation among the powers between iterations and the squared norm of the difference between the current solution and the ground truth, which is obtained through a convex solver. Given that the tolerance $\varepsilon$ is set to $10^{-4}$, the algorithm converges in $I\approx200$ iterations for $K=1$, and in $I\approx50$ iterations for $K=8$. Moreover, it can be seen as, considering this tolerance, the squared norm stays below $10^{-2}$. 

\subsubsection{Narrowband System} The same reasoning used for the wideband system applies to the narrowband one, which is then associated with a complexity of $\mathcal{O}\left(K^3I+K^2MI\right)$ flops for the \glspl{pa} consumed power solution, while $\mathcal{O}\left(K^3+K^2M\right)$ for the transmit power solution.

\subsubsection{Asymptotic Wideband System} In the asymptotic wideband scenario, the precoding matrices have $M_\mathrm{a}$ rows, then are computed with complexity $\mathcal{O}\left(K^3Q+K^2M_\mathrm{a}Q\right)$. Before executing Algorithm~\ref{alg:Madag}, which involves only comparisons, we need to solve the quartic equation~(\ref{eq:quartic_eq}). The solution of quartics can be computed with a fixed number of operations. Calculating the last coefficient in~(\ref{eq:quartic_eq}) requires the computation of $\tr{\mat{D}_{\vect{\beta}}^{-1}\mat{D}_{\vect{\gamma}}}$, which has complexity $\mathcal{O}\left(K\right)$. Therefore, we can consider that retrieving $M_\mathrm{a}^\dag$, which needs to be recomputed when the large-scale fading coefficients change, has lower complexity than retrieving the precoding matrices, which instead need to be retrieved every time the small-scale fading coefficients change. We can therefore consider the complexity to be $\mathcal{O}\left(K^3Q+K^2M_\mathrm{a}Q\right)$. In the conventional case, the complexity remains $\mathcal{O}\left(K^3Q+K^2MQ\right)$. 

Table~\ref{tab:compl} summarizes the complexities of the precoders for the three investigated systems. Given the behavior of Algorithm~\ref{alg:fpim}, the proposed wideband and narrowband precoders have complexities of $1$ or $2$ orders of magnitudes larger than the conventional precoders. However, the larger complexity is justified by the large savings shown in Fig.~\ref{fig:gains}, especially for narrowband systems. Instead, the asymptotic wideband precoder has better performance and lower complexity in low-to-medium load.

%% file: 7_Conclusion.tex
\section{Conclusion}
\label{sec:conclusion}

In this paper, we have studied \gls{mmimo} precoders optimizing consumed power in both narrowband and wideband systems. In narrowband systems not subject to per-antenna power constraints, the optimal precoder activates few \gls{bs} antennas, with consequent savings up to a factor of $3\times$ in the \gls{bs} consumption by deactivating the \gls{rf} chains of the non-active antennas. In wideband systems, more \gls{bs} antennas are progressively activated as the number of subcarriers increases, hence no \gls{rf} chains can be switched off. The asymptotic analysis of wideband systems subject to per-antenna power constraints reveals, however, how a simple optimization on the number of active antennas can lead to \gls{bs} consumption savings up to a factor of $2.8\times$ in low-traffic scenarios. Complexity analysis suggests the viability of the proposed solutions, with the number of active antennas that is only computed when the large-scale fading changes.

A fundamental extension of this work is the analysis in the context of distributed \gls{mmimo}. The power consumption model and \gls{qos} constraints might be redefined, depending on the selected implementation. Relevant aspects to be included in the characterization of cellular systems are the impact of non-ideal channel estimation and the effect of \gls{pa} non-linear distortions. The comparison with hybrid precoding architectures constitutes also an interesting future direction, as well as the consideration of correlation between \gls{bs} antennas in the channel modeling.

%% file: Appendices.tex
\appendices

\section{Proof of Theorem~\ref{th:ZF_pcons_sol}} 
\label{proof1}
The Lagrangian function of problem~(\ref{eq:ZF_pcons}) under $\mathbf{(As1)}$, where the Lagrange multipliers are $\left\{\lambda_{k',k,q} \inC{}\right\}$ and the single element of $\tilde{\mat{D}}_{\vect{\gamma}}^{1/2}\sigma_\nu$ is $d_{k',k}$, is
\begin{equation}
\begin{aligned}
	& \mathcal{L}\left(w_{0,0,0},\dotsc,w_{M-1,K-1,Q-1},\lambda_{0,0,0},\dotsc,\lambda_{K-1,K-1,Q-1}\right) \\
	&= \underbrace{\sum_{m=0}^{M-1}\left(\sum_{k=0}^{K-1}\sum_{q=0}^{Q-1}|w_{m,k,q}|^2\right)^{1/2}}_{(a)}
	+ \underbrace{\sum_{q=0}^{Q-1}\Re\left(\sum_{k'=0}^{K-1}\sum_{k=0}^{K-1}\lambda_{k',k,q}\right.
	\left.\left(\sum_{m=0}^{M-1}h_{k',m,q}w_{m,k,q}-d_{k',k}\right)\right)}_{(b)}.
\end{aligned}
\end{equation}
By using the property that, for $z\inC{}$, $\Re\{{z}\}=\frac{1}{2}\left(z+z^*\right)$, we expand the constraint as
\begin{equation}
\begin{aligned}
(b) &= \frac{1}{2}\sum_{q=0}^{Q-1}\sum_{k'=0}^{K-1}\sum_{k=0}^{K-1}\lambda_{k',k,q}\left(\sum_{m=0}^{M-1}h_{k',m,q}w_{m,k,q}-d_{k',k}\right) \\
	&+ \frac{1}{2}\sum_{q=0}^{Q-1}\sum_{k'=0}^{K-1}\sum_{k=0}^{K-1}\lambda_{k',k,q}^*\left(\sum_{m=0}^{M-1}h_{k',m,q}^*w_{m,k,q}^*-d_{k',k}\right).
\end{aligned}
\end{equation}
Let us now compute the Wirtinger derivative of the Lagrangian function using the properties that, for $z\inC{}$, $\frac{\partial z^*}{\partial z^*}=1$, $\frac{\partial \left(zz^*\right)}{\partial z^*}=z$, and $\frac{\partial z}{\partial z^*}=0$~\cite{Hjorungnes11}. To differentiate $(a)$, we note that, for $z_1,\dotsc,z_N\inC{}$, $\frac{\partial}{\partial z_n^*}\left(\left(\sum_{n=1}^Nz_nz_n^*\right)^{1/2}\right) = \frac{1}{2}z_n\left(\sum_{n'=1}^Nz_{n'}z_{n'}^*\right)^{-1/2}$. The differentiation of $(b)$ follows from the basic properties listed above.
The Wirtinger derivative of $\mathcal{L}$ with respect to a precoding coefficient is then given by
\begin{equation}
\label{eq:app_der}
	\frac{\partial \mathcal{L}}{\partial w_{m,k,q}^*} = \frac{1}{2}\frac{w_{m,k,q}}{\left(\sum_{k'=0}^{K-1}\sum_{q'=0}^{Q-1}|w_{m,k',q'}|^2\right)^{1/2}}
	+ \frac{1}{2}\sum_{k'=0}^{K-1}\lambda_{k',k,q}^*h_{k',m,q}^*.
\end{equation}
Recalling that $p_m = \sum_{k=0}^{K-1}\sum_{q=0}^{Q-1}|w_{m,k,q}|^2$ and setting~(\ref{eq:app_der}) to zero, we obtain
\begin{equation}
\label{eq:app_single}
\begin{aligned}
	\frac{\partial \mathcal{L}}{\partial w_{m,k,q}^*} = 0 \iff w_{m,k,q} = -p_m^{1/2}\sum_{k'=0}^{K-1}\lambda_{k',k,q}^*h_{k',m,q}^*.
\end{aligned}
\end{equation}
The constraint gives
\begin{equation}
\begin{aligned}
	\sum_{m=0}^{M-1}h_{k',m,q}^*w_{m,k,q}^*-d_{k',k} = 0
\end{aligned}
\end{equation}
which, using~(\ref{eq:app_single}), can be rewritten as
\begin{equation}
\begin{aligned}
    \sum_{m=0}^{M-1}h_{k',m,q}^*p_m^{1/2}\sum_{k''=0}^{K-1}\lambda_{k'',k,q}h_{k'',m,q} = -d_{k',k}.
\end{aligned}
\end{equation}
We are left with the following system of equations
\begin{equation}
\label{eq:app_sys}
\begin{cases}
    w_{m,k,q} = -p_m^{1/2}\displaystyle\sum_{k'=0}^{K-1}\lambda_{k',k,q}^*h_{k',m,q}^* \\
	\displaystyle\sum_{m=0}^{M-1}h_{k',m,q}^*p_m^{1/2}\displaystyle\sum_{k''=0}^{K-1}\lambda_{k'',k,q}h_{k'',m,q}=-d_{k',k}
\end{cases}.
\end{equation}
Making use of $\mat{D}_{\vect{p}} = \mathrm{diag}\left(p_0,\dotsc,p_{M-1}\right)$ and denoting with $\mat{\Lambda}_q\inC{K\times K}$ the matrix containing the Lagrange multipliers for the subcarrier $q$, (\ref{eq:app_sys}) can be written in matrix form as
\begin{equation}
\label{eq:app_sys_mat}
\begin{cases}
	\mat{W}_q = -\mat{D}_{\vect{p}}^{1/2}\hr{\mat{H}_q}\mat{\Lambda}_q^*\\
	\mat{H}_q^*{\mat{D}_{\vect{p}}^{1/2}}\tp{\mat{H}_q}\mat{\Lambda}_q = -\tilde{\mat{D}}_{\vect{\gamma}}^{1/2}\sigma_\nu
\end{cases}.
\end{equation}
The matrix $\mat{\Lambda}_q$ is then given by
\begin{equation}
\begin{aligned}
    \mat{\Lambda}_q = -\left(\mat{H}_q^*{\mat{D}_{\vect{p}}^{1/2}}\tp{\mat{H}_q}\right)^{-1}\tilde{\mat{D}}_{\vect{\gamma}}^{1/2}\sigma_\nu
\end{aligned}
\end{equation}
therefore, by substituting it back to the first expression of~(\ref{eq:app_sys_mat}), we obtain
\begin{equation}
\begin{aligned}
\mat{W}_q = \mat{D}_{\vect{p}}^{1/2}\hr{\mat{H}_q}\left(\mat{H}_q\mat{D}_{\vect{p}}^{1/2}\hr{\mat{H}}_q\right)^{-1}\tilde{\mat{D}}_{\vect{\gamma}}^{1/2}\sigma_\nu.
\end{aligned}
\end{equation}

\section{Proof of Theorem~\ref{th:p_PAs_asympt}}
\label{proof2}
Problem~(\ref{eq:ZF_pcons}) is convex. In the asymptotic wideband regime, the power allocations in the neighborhood of the global optimum achieve substantially the same performance. Fig.~\ref{fig:pm-vs-m_MC} illustrates as (i) different power allocations consume approximately the same amount of power while achieving the users' \gls{qos}, and (ii) the power tends to be uniformly allocated among the antennas as $Q$ increases. Among the possible allocations, let us consider a uniform allocation, i.e., $\mat{D}_{\vect{p}}=p\mat{I}_M$. When validating the asymptotic results, we will prove that this allocation is optimal. Equation~(\ref{eq:MC-MU_fpe}) becomes
\begin{equation}
p = \sum_{q=0}^{Q-1}\sum_{k=0}^{K-1}\left|\left[\hr{\mat{H}_q}\left(\mat{H}_q\hr{\mat{H}_q}\right)^{-1}\tilde{\mat{D}}_{\vect{\gamma}}^{1/2}\sigma_\nu\right]_{m,k}\right|^2.
\end{equation}
By summing over the antennas we can write
\begin{equation}
\begin{split}
Mp &= \sum_{q=0}^{Q-1}\sum_{m=0}^{M-1}\sum_{k=0}^{K-1}\left|\left[\hr{\mat{H}_q}\left(\mat{H}_q\hr{\mat{H}_q}\right)^{-1}\tilde{\mat{D}}_{\vect{\gamma}}^{1/2}\sigma_\nu\right]_{m,k}\right|^2 \\
&= \sum_{q=0}^{Q-1}\tr{\hr{\mat{H}_q}\left(\mat{H}_q\hr{\mat{H}_q}\right)^{-1}\tilde{\mat{D}}_{\vect{\gamma}}\sigma_\nu^2\left(\mat{H}_q\hr{\mat{H}_q}\right)^{-1}\mat{H}_q}.
\end{split}
\end{equation}
By using the cyclic property of the trace
\begin{equation}
\label{eq:inproof2}
p = \frac{1}{M}\sum_{q=0}^{Q-1}\tr{\left(\mat{H}_q\hr{\mat{H}_q}\right)^{-1}\tilde{\mat{D}}_{\vect{\gamma}}\sigma_\nu^2}.
\end{equation}
Under $\mathbf{(As2)}$ and recalling that $\tilde{\mat{D}}_{\vect{\gamma}} = \frac{1}{Q}\mat{D}_{\vect{\gamma}}$, the law of large numbers gives
\begin{equation}
\frac{1}{Q}\sum_{q=0}^{Q-1}\left(\mat{H}_q\hr{\mat{H}_q}\right)^{-1}\mat{D}_{\vect{\gamma}}\sigma_\nu^2 \to \expt{\left(\mat{H}_q\hr{\mat{H}_q}\right)^{-1}\mat{D}_{\vect{\gamma}}\sigma_\nu^2}
\end{equation}
then~(\ref{eq:inproof2}) becomes
\begin{equation}
p \to \frac{1}{M}\tr{\expt{\left(\mat{H}_q\hr{\mat{H}_q}\right)^{-1}\mat{D}_{\vect{\gamma}}\sigma_\nu^2}}.
\end{equation}
Under $\mathbf{(As3)}$, the Gram matrix $\mat{H}_q\hr{\mat{H}_q} \sim \cw{\mat{D}_{\vect{\beta}}}{M}{K}$, then its inverse $\left(\mat{H}_q\hr{\mat{H}_q}\right)^{-1} \sim \icw{\mat{D}_{\vect{\beta}}^{-1}}{M}{K}$~\cite{Maiwald97}. We obtain a deterministic expression of the transmit power at each antenna
\begin{equation}
\label{eq:p_tilde}
p \to \frac{1}{M(M-K)}\mathrm{tr}\left(\mat{D}_{\vect{\beta}}^{-1}\mat{D}_{\vect{\gamma}}\sigma_\nu^2\right)
\end{equation}
thereby the total \glspl{pa} consumed power can be computed as
\begin{equation}
p_\mathrm{PAs} = M\alpha p^{1/2} \to \alpha\left(\frac{M}{M-K}\mathrm{tr}\left(\mat{D}_{\vect{\beta}}^{-1}\mat{D}_{\vect{\gamma}}\sigma_\nu^2\right)\right)^{1/2}.
\end{equation}
 
\section{Proof of Lemma~\ref{lemma1}}
\label{proof3}
The optimal number of active antennas lies in between $K+1$ (required by the \gls{zf} constraint) and $M$. By considering a continuous number of active antennas $x\in \mathbb{R}$, the function~(\ref{eq:p_cons,BS_asympt})
\begin{equation}
\label{eq:fx}
f(x) = t\left(\frac{x}{x-K}\right)^{1/2}+p_\mathrm{fix}+\mathcal{C}x
\end{equation}
where $t = \alpha\left(\mathrm{tr}\left(\mat{D}_{\vect{\beta}}^{-1}\mat{D}_{\vect{\gamma}}\sigma_\nu^2\right)\right)^{1/2}$, is convex in the domain $]K,M]$. To prove that, let us compute the first derivative of $f$
\begin{equation}
\label{eq:firder}
\frac{df(x)}{dx} = t\left(\frac{x^{-1/2}}{2\left(x-K\right)^{1/2}}-\frac{x^{1/2}}{2\left(x-K\right)^{3/2}}\right)+\mathcal{C}
= -\frac{t}{2}\left(\frac{Kx^{-1/2}}{\left(x-K\right)^{3/2}}\right)+\mathcal{C}.
\end{equation}
By taking the second derivative
\begin{equation}
\label{eq:secder}
\frac{d^2f(x)}{dx^2} = -\frac{tK}{2}\left(-\frac{x^{-3/2}}{2\left(x-K\right)^{3/2}}-
\frac{3x^{-1/2}}{2\left(x-K\right)^{5/2}}\right)
= \frac{tK}{4}\frac{4x-K}{x^{3/2}\left(x-K\right)^{5/2}}
\end{equation}
we can see that, for $x>K$, (\ref{eq:secder}) is always positive.
To minimize~(\ref{eq:fx}), one can compute its derivative with respect to $x$, given in~(\ref{eq:firder}), set it to zero, and isolate the constant term, obtaining the following polynomial
\begin{equation}
\label{eq:polyn}
x\left(x-K\right)^3-\frac{tK}{2\mathcal{C}}=0.
\end{equation}
From Descartes' rule of signs, (\ref{eq:polyn}) has 2 positive real roots. Among them, there is only one bigger than $K$. This can also be seen by visualizing~(\ref{eq:p_cons,BS_asympt}) in Fig.~\ref{fig:as_example}. The $\min\{\cdot\}$ and $\max\{\cdot\}$ operations guarantee that the final value stays within the allowed range. Due to the convexity of the function~(\ref{eq:fx}), the result of the ceil-floor operation $\left\lfloor\cdot\right\rceil$ is guaranteed to be the optimal integer value.


\section{Proof of Theorem~\ref{th:Ma_dag}}
\label{proof4}
Let us consider problem~(\ref{eq:p_cons,BS_papc_gen}). Fixing the value of $M_\mathrm{a}$, the minimization problem with respect to the precoding coefficients is
\begin{subequations}
\label{eq:internal}
\begin{align}
    \minimize_{\substack{\{w_{m,k,q}\}\\m=0,\dotsc,M_\mathrm{a}-1}} \quad\quad 	    & 	 p_\mathrm{BS}=p_\mathrm{PAs}+p_{\mathrm{fix}}+\mathcal{C}M_\mathrm{a} \label{eq:internal1}\\
    \mathrm{\;\;subject\ to} \quad\quad                & 	\mat{H}_q\mat{W}_q = \tilde{\mat{D}}_{\vect{\gamma}}^{1/2}\sigma_\nu \;\; \forall q, \label{eq:internal2}\\
    & p_m\leq p_{\mathrm{max}}\;\; \forall m. \label{eq:internal3}
\end{align}
\end{subequations}
From Theorem~\ref{th:p_PAs_asympt}, we know that the solution to the above problem (in the asymptotic wideband regime and \gls{iid} Rayleigh fading) when the per-antenna power constraints are not binding tends to allocate the power uniformly among the antennas, 
\begin{equation}
\label{eq:pm}
    p_m \to \overline{p} = \frac{1}{M_\mathrm{a}(M_\mathrm{a}-K)}\mathrm{tr}\left(\mat{D}_{\vect{\beta}}^{-1}\mat{D}_{\vect{\gamma}}\sigma_\nu^2\right) \quad \forall m.
\end{equation}
Let us consider two cases, recalling that we are considering an arbitrary value of $M_\mathrm{a}$:
\begin{enumerate}[a)]
	\item $\overline{p}\leq p_\mathrm{max}$, which guarantees that the above solution is optimal. Indeed, solution~(\ref{eq:pm}) minimizes~(\ref{eq:internal1}) while satistying~(\ref{eq:internal2}), and in this case it fulfils also~(\ref{eq:internal3}).
	\item $\overline{p}>p_\mathrm{max}$, which indicates that the above solution is not feasible by using $M_\mathrm{a}$ antennas. Indeed, to satisfy the constraints~(\ref{eq:internal3}), we would need to decrease every value of $p_m$ to $p_\mathrm{max}$ or a smaller quantity. By doing this, we would need a solution consuming less power (given that \glspl{pa} consumption increases monotonically with $p_m$) while fulfilling~(\ref{eq:internal2}). This is not possible, otherwise this new solution would correspond to the optimal one for the unconstrained problem, i.e., the one allocating the power $\overline{p}$ to all antennas.
\end{enumerate}
The above reasoning entails that, in the asymptotic wideband regime and \gls{iid} Rayleigh fading, the per-antenna power constraints imply a lower bound on the number of active antennas $M_\mathrm{a}$ (by using~(\ref{eq:pm}) and assuming $M_\mathrm{a}>K$)
\begin{equation}
\overline{p}\leq p_\mathrm{max} \iff M_\mathrm{a}^2-KM_\mathrm{a}-\frac{\mathrm{tr}\left(\mat{D}_{\vect{\beta}}^{-1}\mat{D}_{\vect{\gamma}}\sigma_\nu^2\right)}{p_\mathrm{max}} \geq 0,
\end{equation}
\begin{equation}
\label{eq:Ma_min_papc}
M_\mathrm{a} \geq \frac{1}{2}\left(K+\left(K^2+\frac{4\tr{\mat{D}_{\vect{\beta}}^{-1}\mat{D}_{\vect{\gamma}}\sigma_\nu^2}}{p_\mathrm{max}}\right)^{1/2}\right).
\end{equation}
Using this lower bound as a constraint in the original problem, we can drop the per-antenna power constraints and solve the equivalent problem
\begin{equation}
\label{eq:p_cons,BS_papc}
\begin{aligned}
    & \minimize_{M_a}  \quad  	    && 	\overline{p_\mathrm{BS}} = \alpha\left(\frac{M_a}{M_a-K}\mathrm{tr}\left(\mat{D}_{\vect{\beta}}^{-1}\mat{D}_{\vect{\gamma}}\sigma_\nu^2\right)\right)^{1/2}+ p_{\mathrm{fix}}+\mathcal{C}M_a\\
    & \mathrm{subject\ to} \quad    && M_a \geq \frac{1}{2}\left(K+\left(K^2+\frac{4\tr{\mat{D}_{\vect{\beta}}^{-1}\mat{D}_{\vect{\gamma}}\sigma_\nu^2}}{p_\mathrm{max}}\right)^{1/2}\right), \\
    &                               && M_a\leq M.
\end{aligned}
\end{equation}
The condition in $\mathbf{(As4)}$ ensures that the power constraint is always satisfied by activating all the $M$ antennas. It can be simply derived from~(\ref{eq:pm}) by setting $M_\mathrm{a}=M$.